\documentclass[11pt]{article}           
\usepackage{latexsym,theorem,epsfig}

\newcommand{\mmul}{}

\newcommand{\lf}{\left}
\newcommand{\rg}{\right}
\newcommand{\lgg}[1]{\lg\lg #1}  

\newcommand{\ceil}[1]{\lceil #1 \rceil}
\newcommand{\floor}[1]{\lfloor #1 \rfloor}
\newcommand{\Ceil}[1]{\left\lceil #1 \right\rceil}

\newcommand{\mx}[1]{\max{\{ #1 \}}}

\newcommand{\mn}[1]{\min{\{ #1 \}}}

\newenvironment{elist}{\begin{list}%
                           {(\arabic{itemno})} {\usecounter{itemno}
                               \setlength{\itemsep}{.1ex} 
                               \setlength{\rightmargin}{\leftmargin}}
                            }{\end{list}} \newcounter{itemno} 
\newenvironment{ilist}{\begin{list}%
                           {{\bf $\bullet$ }} {
                               \setlength{\itemsep}{.1ex}
                               \setlength{\rightmargin}{\leftmargin}}
                            }{\end{list}}

\newcommand{\BBox}{\rule{0.1in}{0.1in}}
\newenvironment{proof}{\noindent {\thheadfont Proof:}}{\BBox \\*}

\newcommand{\bsplib}{\textsl{BSPlib}}

\newcommand{\thheadfont}{\bf}

\newtheorem{lemma}{Lemma}[section]
\newtheorem{corollary}{Corollary}[section]

\newtheorem{proposition}{Proposition}[section]
\newtheorem{claim}{Claim}[section]

\theoremheaderfont{\thheadfont}
\theorembodyfont{\thbodyfont}

\begin{document}
\title{BSP Sorting: An experimental Study \\
      }
\author{ Alexandros V. Gerbessiotis \\
         CS Department \\
         New Jersey Institute of Technology \\
         Newark, NJ 07102. \\
                 \and
        Constantinos  J. Siniolakis       \\
        The American College of Greece \\
        6 Gravias St. \\
        Aghia Paraskevi\\
        Athens, 15342\\
        Greece.
       }

\maketitle
\thispagestyle{empty}

\begin{abstract}
The Bulk-Synchronous Parallel model of computation has been used for the 
architecture independent design and analysis of parallel algorithms whose
performance  is expressed not only in terms of problem size $n$ but also 
in terms of parallel machine properties. 
In this paper  the performance of implementations of deterministic and 
randomized BSP sorting algorithms is examined. 
The deterministic algorithm is the one in \cite{GS96a,GS99a} that
uses deterministic regular oversampling and parallel sample sorting and is
augmented to handle duplicate keys transparently with optimal asymptotic
efficiency. The randomized algorithm derives from the algorithm in 
\cite{GV94} that is sample-sort based and uses oversampling
and the ideas introduced with the deterministic algorithm. The resulting
randomized design, however, works differently from traditional parallel 
sample-sort based algorithms and is also augmented to transparently
handle duplicate keys with optimal asymptotic efficiency
thus eliminating the need to tag all input keys and
to double communication/computation time.
Both algorithms are shown to balance the work-load evenly among the
processors and the use and precise tuning of oversampling that the
BSP analysis allows combined
with the transparent duplicate-key handling insures regular and
balanced communication.
Both algorithms  have been
implemented in ANSI C and their performance (scalability and 
efficiency issues)
has been studied on a distributed memory machine, a Cray T3D. 
The experimental results obtained from these implementations are 
reported in this work. The validity of the theoretical model is 
also tested; based on  the theoretical performance of 
each algorithm under the BSP model and the BSP parameters 
of a Cray T3D, it is  possible to 
estimate the actual performance of the implementations based on 
the theoretical performance of the designed algorithms.
\end{abstract}

\newpage

\pagenumbering{arabic}

\section{Introduction}
\label{s1}

%In the past decades sequential computing has enjoyed unprecedented 
%levels of growth and popularity. Even though the limitations of 
%sequential computing have become increasingly evident,
%the field of parallel computing lagged behind. 
%The importance of parallel computing, however, has grown out of 
%traditional fields
%to include new ones, as evident, among others, in the effective design 
%of new drugs, the optimization of engineering designs for fuel-efficient
%vehicles and aircraft, the understanding of natural phenomena, weather 
%forecasting, climate and air-pollution modeling, and the evaluation of 
%the safety and effectiveness of nuclear weapon stockpiles.
One of the early attempts to model a parallel computer has been
the Parallel Random Access Machine (PRAM)  \cite{Fortune78} which is one
of the most widely studied abstract models of a parallel computer.
The advantage of the PRAM model in terms of algorithm design is its
simplicity in capturing parallelism and abstracting away the communication
requirements of parallel computing such as communication latency, block
transfers, memory and network conflicts during routing, bandwidth of
interconnection networks, interprocessor communication, memory management
and synchronization.
In addition, most of the algorithms designed for the PRAM model work
under unlimited parallelism assumptions, where the number of processors
utilized to solve a problem is a function of  problem size.
It is these advantages that make the PRAM easy
to program and popular for algorithm development.  The success of the
PRAM model is witnessed by the hundreds of algorithms that have been 
designed and extensively analyzed on various variants of the PRAM.

In parallel machines that have been built so far, however, communication
is expensive, synchronization is not instantaneous and unlimited parallelism
is unavailable.
The majority of parallel machines and available parallel software do not allow
programmers to write programs that are both efficient, portable and scalable;
such objectives are only feasible by fine-tuning parallel code and taking
into account the finest details of the underlying architecture and the
programming interface that is available to the programmer.
Parallel algorithm design, on the other hand, usually ignores
communication and/or synchronization issues and works only under 
unlimited parallelism assumptions (or by simulations on limited processor
models). Many parallel algorithms are thus 
designed having in mind factors not present in any
reasonable machine, such as zero communication delay or infinite bandwidth.

The need for architecture independent parallel algorithm design is thus
immediate for effective and efficient programming of existing and future
parallel hardware platforms. Such an algorithm design would work in two
steps. First, an abstraction of parallel hardware is obtained to allow
the performance of a parallel machine to be described in some general
architecture independent terms through a number of parameters that
reflect the computation, communication, and synchronization capabilities
of the hardware. Then, the performance of a parallel
algorithm would be expressed as a function of these
parameters and problem size similarly to traditional sequential 
algorithm design.
This way it would be possible to describe the performance of a parallel
algorithm not only on existing machines but also on machines not currently
built as long as  sufficient information on the not-yet-built machines
can be obtained (i.e. the parameters of the abstraction become available).

The introduction of realistic parallel computer models
such as the Bulk-Synchronous Parallel (BSP) model \cite{Valiant90,Valiant90b}
and  the LogP~\cite{Culler} model and their extensions (\cite{EBSP,Baumker95})
come to address these limitations of
parallel computing. Both models take into reasonable consideration
synchronization, communication, and  latency issues related to
both communication, bandwidth limitations and network conflicts during routing.

A  considerable amount of work has been devoted in the study of the BSP
model, and the  analysis and design of BSP algorithms.  In this work, we
present parallel implementations, using the Oxford BSP Toolset, \bsplib{} 
\cite{Hill96}, of deterministic and randomized BSP sorting
algorithms on a distributed memory machine, a Cray T3D. A report 
is also included on the performance and relative merits of these algorithms 
and the knowledge gained from the study of the predicted theoretical 
performance and its comparison to actual running time.
A comparison of our algorithm implementations to other parallel
sorting algorithm implementations is also presented.

\subsection{The Parallel Programming Model}

The {\em Bulk-Synchronous Parallel} (BSP) model of computation deals explicitly 
with the notion of communication and synchronization among computational 
tasks and has been proposed in \cite{Valiant90,Valiant90b} as a 
unified framework for the design, analysis and programming of general purpose 
parallel computing systems.  It offers the prospect of achieving both
scalable parallel performance and architecture independent portable and
reusable parallel software 
and provides a framework which permits the performance of parallel software
to be analyzed and predicted in a consistent and concise 
way. 
%The term {\em bulk-synchronous} reflects the underlying position of the 
%model between the two extremes, 
%(i) entirely synchronous systems (a global
%clock is utilized to achieve synchronization on every instruction), and (ii)
%fully asynchronous systems (synchronization is realized pairwise whenever
%communication is required). 
Processor components in the BSP model advance 
jointly through the program, with the required remote communication 
occurring between {\em supersteps}. 
A superstep may be thought of as a segment of
computation during which each processor performs a given task using
data local to the processor before the start of the superstep.
Such a task may include  local computations, 
message transmissions, and  message receipts.  
A BSP computer as described in 
\cite{Valiant89,Valiant90,Valiant90b} then consists of the following three 
components:
(a) a {\em collection of} $p$ {\em processor/memory components},
(b) a {\em communication network} that delivers messages point to 
   point among the components,
and (c)  {\em facilities} for synchronization
   of all or a subset of the processors.

Any BSP computer can be modeled by $p$, the number of processor components,
$L$, the minimal time, measured in terms of basic computation steps,
between successive synchronization operations, and $g$ the ratio of
the total throughput of the whole system in terms of basic computation
steps, to the throughput of the router in terms of words of
information delivered.  The  definition of $g$ relates to the 
routing of an $h$-relation in continuous message usage, that is, 
the situation  where each processor sends or receives at most $h$ messages 
(or words of information); $g$ is the cost of communication so that
an $h$-relation is realized within  $gh$ steps \cite{Valiant90},
for any $h$ such that $h \geq h_0$, where $h_0$ is a machine
dependent parameter. Otherwise, if $h < h_0$, the cost of communication 
is $L$.
%It is noted that, although BSP stresses global barrier style
%synchronization, pairs of processor components may synchronize pairwise 
%by sending messages to and from an agreed memory location.  However, such 
%exchanges should respect the superstep rules \cite{GV94}.
%
%For any particular machine the minimum time necessary for doing
%synchronization imposes a lower bound on $L$,
%and the theoretical bound on $g$ may be realizable only for large 
%enough $L$, so that the communication network can achieve 
%the stated capacity.  
%The maximum size of  a superstep is a property of the 
%application, since by definition a superstep can accommodate as much work 
%as can be done without a remote data transfer. The actual values of 
%$L$ and $g$ for any BSP machine are likely to depend on $p$. In 
%particular, the parameters $L$ and $g$ are considered to be 
%non-decreasing functions of 
%the number of processors used by an algorithm. 
%A detailed discussion of the BSP model can be
%found in \cite{Valiant90}.

A superstep can complete at any time after $L$ time units.
The time complexity of a superstep {\cal S} in a BSP algorithm is determined 
as follows.  Each superstep is charged $\mx{L, x + gh}$ basic time steps, 
where $x$ is the maximum number of basic computational operations executed 
by any processor during {\cal S}, and $h$ is the maximum amount of 
information transmitted or received by any processor.  
\label{UnifCharging}

The performance of a BSP algorithm $\cal{A}$ is specified in three parts.  
First, a sequential algorithm $\cal{A}^{*}$ is specified with which the
BSP algorithm is compared  and  the charging policy for basic operations 
of both $\cal{A}$ and $\cal{A}^{*}$ is made explicit.  
Ratios between runtimes on pairs of models that have the 
same set of local instructions will  thus be measured and therefore 
operations can be defined in a higher level of abstraction than 
machine level instructions.  

Second, two ratios $\pi$ and $\mu$ are specified.  The former,
$\pi$, is the ratio between the  computation time $C_{\cal{A}}$, of the 
BSP algorithm, over the time $C_{\cal{A}^{*}}$ of the comparing sequential
algorithm divided by $p$, i.e., $\pi = pC_{\cal{A}}/C_{\cal{A}^{*}}$,
and is a measure of the computational efficiency of $\cal{A}$.
The latter, $\mu$, is the ratio between the parallel time 
$M_{\cal{A}}$ required by the communication supersteps of the BSP 
algorithm and the  computation time of $\cal{A}^{*}$ divided by $p$, i.e., 
$\mu = pM_{\cal{A}}/C_{\cal{A}^{*}}$, and is a measure of the impact
of communication in achieving optimal efficiency (the lower the $\mu$
the more communication efficient $\cal{A}$ is).  
The ratio $p/(\pi+\mu)$
gives the speedup, 
and the ratio $1/(\pi+\mu)$ the parallel efficiency of algorithm $\cal{A}$.

Finally, conditions on $n$, $p$, $L$ and $g$ are then specified that 
are sufficient for the algorithm to make sense and the 
claimed (upper) bounds on $\pi$ and $\mu$ to hold.  Corollaries are claimed
for some sufficient conditions for the most interesting optimality 
criteria, such as $c$-optimality, i.e., $\pi = c+o(1)$ and $\mu = o(1)$.  
Insistence on one-optimality leads to algorithms
that only work for restricted ranges of $p$, $L$ and $g$.
All asymptotic bounds refer to the problem size as $n \rightarrow \infty$.  
The other parameters may also tend to $\infty$ if they are expressed 
in terms of $n$, though it is not assumed so. 

The charging policy for local operations on the BSP model for the
sorting algorithms is described. Since the mechanisms 
of the BSP model are similar to those used by sequential models and 
performance ratios of these two are important, operations can be defined
in a higher level than machine level instructions.  
A dominant term in the BSP
algorithms to be described later is contributed by sequential sorting
performed independently at each processor.  A  charge of $n \lg{n}$ time 
for sorting $n$ keys sequentially \cite{Knuth73}, and
$n \lg{q}$ time for merging $q$ lists of total size $n$ 
\cite{Knuth73} are claimed.  
For other operations, 
$O(1)$ time is charged for operations over an associative operator in 
parallel-prefix computations and $O(1)$ time for each comparison performed.
Finally, $\ceil{\lg{n}}$ comparisons are charged for performing 
binary search in a sorted sequence of length $n-1$.

\section{Sorting on the BSP model: An overview}
\label{s2}

A significant amount of work has been devoted in the study of the BSP
model 
\cite{Valiant89,Valiant90,Valiant90b,McColl93,McColl94,McColl94b,EBSP,Baumker95}
and the analysis 
\cite{Baumker95,Bilardi96,Bisseling93,GV94,GS96a,Goodrich96} 
and design of BSP algorithms 
\cite{Baumker96,Baumker96a,GS96b,GS96c,GS96d,GS96e,GS96f,McColl93,McColl94,McColl94b}. 

The problem of parallel sorting $n$ keys \cite{Leighton91} has drawn 
considerable attention.  
First approaches were sorting networks introduced by Batcher 
\cite{Batcher68}.  Since then, a wealth of literature 
concerning parallel sorting has been published \cite{Leighton85, Leighton91}.  
The first sorting network to achieve, however, the optimal 
$O(\lg{n})$ time bound to sort $n$ keys is the AKS \cite{Ajtai83} 
network of width $O(n)$.  
Subsequently,  
Reif and Valiant \cite{Reif87} presented a
randomized sorting algorithm
that runs on a fixed-connection network and 
achieves the same time bound as the AKS algorithm; the constant multipliers in
the randomized algorithm are, however, considerably smaller. 
Cole \cite{Cole88} was the first to present optimal $O(\lg{n})$ time 
sorting algorithms for $n$-processor CREW and EREW PRAMs.  
These methods fail, however, to deliver 
optimal performance when directly implemented on the BSP model 
of computation. 

{\em Note}: In the expressions for running time on the BSP model below, terms
that involve $L$ are included in the communication -- and not 
in computation -- time, to simplify exposition. Polynomial slack would
mean that $n_p =n/p$ is a constant degree polynomial, i.e. $p=n^{1-t}$ for
some constant $t$ such that $0 < t < 1$.

Previously known results on BSP sorting include deterministic 
\cite{Adler95,Bilardi96,GS96a,GS96f,Goodrich96,S96-TR-09} and randomized 
\cite{GV94,GS96c,GS97b} algorithms.

%%AKS
An adaptation of the AKS sorting
network \cite{Ajtai83} on the BSP model is shown in
\cite{Bilardi96} to yield computation and communication time $O(n_p
 \lg{n})$ and $g  O(n_p  \lg{p})+ L O(\lg{p})$ 
respectively.  The constant factors hidden in this  algorithm, 
however, is  considerably large and not of practical significance.

%%ADLER
It is shown in \cite{Adler95} that for all values of $n$
and $p$ such that $n \geq p$, a BSP version of column-sort
\cite{Leighton85} requires $O(\zeta^{3.42} n_p \lg{n_p})$ and 
$O(g \zeta^{3.42} n_p) + O(L \zeta^{3.42} )$
time for computation and communication respectively, where
$\zeta = \lg{n}/\lg{(n/p)}$. 

%%BILARDI et al
In \cite{Bilardi96}, it is shown that an adaptation on the BSP model of the
cube-sort algorithm \cite{Cypher92} requires  computation time 
%$O((25^{\lg^{*}{\!n} - \lg^{*}{\!(n/p)}})  
%(\lg{n}/\lg{(n/p)})^{2}  ((n/p) 
%\lg{(n/p)}))$ time for computation and 
%$g  O((25^{\lg^{*}{\!n} - \lg^{*}{\!(n/p)}})
%\ (\lg{n}/$ $\lg{(n/p)})^{2} \ (n/p))
%+ L  O((25^{\lg^{*}{\!n} - \lg^{*}{\!(n/p)}}) (\lg{n}/$ $\lg{(n/p)})^{2})$ 
$O(25^{\lg^{*}{\!n} - \lg^{*}{\!(n_p)}}  
\zeta^{2}   n_p \lg{n_p}))$ and
$  
O(g \; 25^{\lg^{*}{\!n} - \lg^{*}{\!(n_p)}}
\zeta^{2}   n_p ) 
+   O(L \; 25^{\lg^{*}{\!n} - \lg^{*}{\!(n_p)}} 
\zeta^{2}   )$ for communication.
Furthermore, as noted in \cite{Goodrich96},
a modification of cube-sort 
shown in \cite{plaxton}  eliminates the term
$25^{\lg^{*}{\!n} - \lg^{*}{\!(n_p)}}$ and the resulting algorithm
improves upon the algorithm of \cite{Adler95}.

%%%JPDC94
In \cite{GV94}, a randomized BSP sorting algorithm that uses the
technique of oversampling is introduced that 
for an ample range of the parameters $n$ and $p$, i.e., $p =
O(n/\lg^{1+\alpha}{n})$, for any constant $\alpha >0$, 
requires -- under realistic assumptions -- computation
and communication time $(1+o(1)) \ (n \lg{n}/p)$ 
and $O(g \ \zeta \ n_p ) +
O(L \ \zeta )$ 
respectively, with
high-probability. A similar algorithm for the case $p^2 <n$, but 
with a less tighter analysis is discussed in \cite{SPAA94}. 
Given any $O(n \lg{n})$ worst-case running time sequential sorting 
algorithm, the randomized BSP algorithm utilizes the sequential
algorithm for local 
sorting and exhibits parallel efficiency which is asymptotically 100\% 
provided that the BSP parameters are appropriately bounded (and these
bounds accommodate most, if not all, currently built parallel machines).
The constants hidden affect low order terms and are well defined.
This has been one of the first attempts to establish in some general
way parallel efficiency for an algorithm by taking into consideration 
communication and synchronization delays and describing optimality criteria 
not in absolute terms, that may be meaningless, but by comparing the 
proposed parallel algorithm to the best sequential one. The bounds
on processor imbalance during sorting are tighter than any other 
random sampling/oversampling algorithm  
\cite{HC,Frazer70,Reif87,Reischuk85}.
The randomized algorithm has been implemented on realistic machines
(SGI Power Challenge, Cray T3D and IBM SP2) and exhibited performance 
well within the bounds indicated by the theoretical analysis 
(\cite{Alex96a,GS97a}). We note that sample-based sorting algorithms
are extensively studied and implemented in the literature 
\cite{Blelloch91,Hightower92}. For example, the sample-sort
algorithm in \cite{Blelloch91} is computationally equivalent to the
one round case of \cite{GV94}. The analysis of \cite{GV94} is however
much tighter. It allows better control of oversampling with the end
result of achieving smaller bucket expansion (key imbalance during
the routing operation). The work of \cite{Hightower92} applied to
$d$-dimensional meshes is similar to the general 
algorithm of \cite{GV94}.
%(

%%%IPPS+NEWJOURNAL
In \cite{GS96f,GS97b}, it is shown that  asymptotical
performance comparable to that of \cite{GV94} can be achieved 
for values of $p$ much closer to $n$, i.e., 
$p =\omega(n\lg{\lg{n}}/\lg{n})$ by a new randomized sorting algorithm.
The failure probability is, however, significantly lower. As the 
improvements 
of this algorithm are interesting only for large $p$ and the corresponding conditions 
on $L$ and $g$ for large $p$ can be satisfied only by machines with very 
low $L$ and $g$ (i.e. PRAM-like behaving machines), the practical
implications of this algorithm are less significant and
interesting than its theoretical merit.

%%%MERGE WITH SPAA96 and DETSORT BELOW
Under the same realistic assumptions as \cite{GV94}, 
a new deterministic sorting algorithm is introduced in \cite{GS96a},
and a bound of 
%$(1+(\floor{
%(\frac{\lg{(n/p)}}{\lg{n}})^{-1}
%(1-\frac{\lg{(n/p)}}{\lg{n}}) }
%((1-\frac{\lg{(n/p)}}{\lg{n}})/2)+o(1)) (n \lg{n}/p)$ and 
%$O(g \frac{\lg{n}}{\lg{(n/p)}}\frac{n}{p})+ O(L\;\frac{\lg{n}}{\lg{(n/p)}}) $ 
%is shown for computation and communication respectively, 
 $(1+(
\floor{\zeta  (1- \zeta^{-1} ) }
((1- \zeta^{-1} )/2)+o(1)) (n \lg{n}/p)$ and
$O(g \zeta n_p )+ O(L\; \zeta ) $
is shown for computation and communication respectively,
thus improving upon the upper bounds of 
\cite{Adler95,Ajtai83,Bilardi96,Cypher92,Leighton85}. The bound on
computation is  subsequently improved in \cite{GS96f,GS99a}. 
In particular, for $p = n/\lg^{2+\alpha}{n}$, 
$\alpha = \Omega(1)$, the improved deterministic sorting algorithm 
in \cite{GS96f,GS99a} requires computation and communication time
$(1+2/\alpha+o(1)) \ n\lg{n}/p$ and 
%$O(g \; \frac{\lg{n}}{\lg{(n/p)}} \; \frac{n}{p})+ 
%O(L\; \frac{\lg{n}}{\lg{(n/p)}})$ respectively.
$O(g \; \zeta \; n_p )+
O(L\; \zeta )$ respectively.
The algorithm performs deterministic sorting by extending 
regular sampling, a technique introduced in \cite{SH}, to perform deterministic 
regular oversampling. Past results that use regular sampling 
have been available for cases with $p^2 < n$. The BSP 
algorithm in \cite{GS96a,GS96f,GS99a}
further extends the processor range and
achieves asymptotically optimal efficiency  for a range of $n/p$ that 
is very close to that ratio of the randomized BSP algorithms in
\cite{GV94,GS97b}. 
This was made possible by a detailed and precise quantification of the 
input key imbalance among the processors during the phases of deterministic 
sorting thus contributing to the understanding of regular oversampling. By 
using regular oversampling in a deterministic setting, it is possible to 
regulate the oversampling ratio to bound the maximum key imbalance among the 
processors. As in the case of 
randomized sorting, the insistence, under an architecture independent 
algorithm design, of satisfying the criterion of {\em one-optimality} 
(i.e. optimal speedup/efficiency) led to these improvements.
In this work we use the structure of this algorithm to derive
a sample-based randomized sorting algorithm that does not follow
the standard pattern of sample and splitter select, key routing and
local sorting but instead follows the pattern of the deterministic
algorithm i.e. local sort, sample and splitter select, key routing
and local merging.

%In particular, the randomized algorithm incurs a step where many messages
%of small size are dispatched resulting in fine grain communication; 
%this can be avoided however by performing a sequential integer sort 
%algorithm for an extra cost that although it is linear, its hidden 
%constant is significant. The BSP model assumes that for such fine 
%grain communication a combining of messages ought to be performed and the
%cost model assumes that this is performed at no extra cost; the test 
%platforms were not, however, typical BSP machines.
%The deterministic algorithm on the other hand, does not suffer 
%by such a problem. 
%Observing the differences between the two algorithms, it was then possible 
%to derive a variant of the randomized algorithm in \cite{GV94} 
%based on the deterministic one in \cite{GS96a} that would perform better 
%than both algorithms.

In \cite{Goodrich96, S96-TR-09} independent BSP adaptations of parallel 
merge-sort \cite{Cole88} are presented that for all values of $n$ 
and $p$ such that $p \leq n$, require computation and 
communication/synchronization time 
$O(n \lg{n}/p)$ and 
$O(g \zeta n_p  )+
O(L \zeta )$ 
respectively.
As the BSP adaptation of parallel merge-sort is at least as involved 
as the original parallel merge-sort (EREW PRAM variant) algorithm 
\cite{Cole88}, the constant factors involved in the analysis are 
considerably large and the algorithm seems to be of little practical 
use as opposed to the algorithms in \cite{GV94,GS96a,GS99a}.

\section{Contents of the paper}
\label{s3}

%%%IMPLEMENT

In this work we study an implementation of the deterministic algorithm 
in \cite{GS96f,GS99a} that uses deterministic regular oversampling,
which extends the notion of regular sampling of \cite{SH}, 
and parallel sample sorting that allows the algorithm to
work for a wider range of processor size $p$. The deterministic
algorithm is also augmented to handle transparently duplicate keys 
with optimal asymptotic efficiency. Our method of duplicate-key handling
tags only a fraction $o(n)$ of the $n$ input keys, and does not
double communication and/or computation time that other duplicate
handling approaches may require \cite{jaja1,jaja2,jaja3}. The idea 
of using deterministic regular-oversampling in a 
deterministic sorting algorithm 
and the accurate quantification of its effects by measuring constants
under the BSP model, 
results in a fine and quantifiable balance of computational load and 
communication time.  Combined with our transparent and efficient 
duplicate key handling it leads to an algorithm that maintains its 
optimal performance even if all keys are the same.

The basis of the randomized algorithm is the foundation of randomized
sample-sort algorithms \cite{HC,Frazer70,Reif87,Reischuk85}
that use oversampling \cite{Reif87}
and in particular the
BSP adaptation discussed in \cite{GV94} combined with the ideas used in
the deterministic algorithm. The resulting randomized sorting algorithm
whose 
implementation is studied works differently from other sample-sort
based algorithms including the one in \cite{GV94}. The algorithm
first local sorts, then
sample and splitter selects, routes keys, and in the very end $p$-way
merges sorted sequences as opposed to traditional sample-sort algorithms
that sample and splitter select, route keys, and local sort.
Our approach augmented with the transparent and efficient duplicate-key
handling also used in the deterministic algorithm results in improved
parallel performance, work balance and communication minimization.

Even though the two algorithms look similar, randomized oversampling
is provably superior to deterministic regular-oversampling. The oversampling
parameter in the former case can vary more widely than the corresponding
one of the latter case thus resulting in more balanced communication
and balanced work-load among the processors.

Parallel sorting algorithms originally designed for other models of 
computation (e.g. bitonic-sort \cite{Batcher68,GS96po}) are also 
implemented. The implementation platform is a distributed memory 
machine, a Cray T3D, and the communication library used is the 
Oxford BSP Toolset \cite{Hill96}.

The introduced algorithms are designed and analyzed in an architecture
independent setting. 
Such a design is helpful in deciding the best way to implement
various parallel operations optimally on the given platform based
on the knowledge of the BSP parameters of the target platform only.
For example, problem size $n$ of the experiments, and
machine configuration parameters as expressed in terms of the BSP
parameters $p, L$ and $g$  determine the choice
of broadcasting and parallel prefix algorithms in parallel sorting.
It is also used in determining the choice of values of certain parameters
of the algorithm implementations (such as those related to oversampling
issues).  This choice of parameter values subsequently affects
the maximum input key imbalance among the processors which has a
significant impact on running time.
Insistence on one-optimality both in the design and analysis of the 
given algorithms is  hoped to lead to efficient implementations 
that would confirm the claims of the theoretical analysis.

The theoretical analysis of the algorithms can be used in assessing the 
performance of the implementations because the designed algorithms
have been shown to be  one-optimal and this way, they have been 
directly compared to the best sequential implementations. 
In addition, in the theoretical analysis no hidden constants are involved
with the most significant terms of parallel running time and the
constants in low order terms are well understood.
At the conclusion of the experiments theoretical performance is compared
to observed performance on the test platform for this purpose.

In Section \ref{s4} we introduce some primitive operations that will be used
in the sorting algorithms. In Section \ref{s5} we introduce the deterministic
algorithm  {\sc Sort\_Det\_BSP} that will be implemented.
The algorithm is one of several cases of a more general algorithm described 
in \cite{GS99a}. The performance of the implemented algorithm is however 
analyzed in detail and more accurately than the general algorithm of 
\cite{GS99a}, and experiments on its implementation are discussed in
Section \ref{s6}. Its augmentation to handle duplicate keys transparently
and with optimal asymptotic efficiency is also discussed separately in
Section \ref{duplicate}.
Subsequently, a special case of the  randomized BSP algorithm in \cite{GV94} 
is introduced, called {\sc Sort\_Ran\_BSP} that is a sample-sort based
algorithm. 
The randomized algorithm of our implementations, {\sc Sort\_IRan\_BSP},
derives from {\sc Sort\_Det\_BSP} and {\sc Sort\_Ran\_BSP} and is introduced
in the same section; it works differently however from traditional
sample-sort based algorithms.  Duplicate-key handling is also implemented
transparently for {\sc Sort\_IRan\_BSP} using the technique of
Section \ref{duplicate} used for {\sc Sort\_Det\_BSP}.
For each of the two implemented algorithms, two variants are 
studied experimentally depending on how sequential sorting is performed.
Finally the obtained experimental results are discussed in detail
in Section \ref{s6} and compared to other parallel sorting implementations like 
those in \cite{Hilla,jaja1,jaja2,jaja3}.

The results we obtained justify the belief that architecture independent
parallel algorithm design and analysis is possible, plausible, reliable and 
consistent and can be used to model and predict the performance of
parallel algorithms on a variety of parallel machines with satisfactory
accuracy. The vehicle for such an architecture independent parallel
algorithm design and analysis has been the BSP model of computation.
Its parameters seem to model a parallel computer well enough to make
parallel program performance estimation plausible for such diverse problems
as matrix computations (\cite{G98}) and parallel sorting. 
The experiences related to the latter problem are described in more
detail in this work.

\section{Primitive Operations}
\label{BasicOper}
\label{s4}

In this section  BSP algorithms for two
fundamental operations, namely {\em broadcast} and {\em parallel-prefix} and
{\em parallel radix-sort} are introduced.  All these primitives are auxiliary routines for
the algorithms described in later sections. 
A fully detailed account of such operations along with 
additional results can be found in \cite{GS96e}.

\begin{lemma}
\label{Broadcast}
There exists a BSP algorithm for broadcasting an $n$-word message that
requires time at most $M^{n}_{\mathit{brd}}(p) =(\ceil{n/\ceil{n/h}}+h-1)
\mx{L, gt\ceil{n/h}}$,  for any integer $2 \leq t \leq p$, where 
$h=\ceil{\log_{t}{((t-1)p+1)}}-1$.
\end{lemma}

\begin{proof}
The underlying algorithm employs a pipelined $t$-ary tree, consisting of $p$ 
nodes, for some appropriate $t$. The depth of the tree is 
$h=\ceil{\log_{t}{((t-1)p+1)}}-1$. In each superstep, the root processor 
of a subtree
sends $t$ copies of a separate $m$-word segment of the original message to 
its children,  where $m = \ceil{n/h}$. Similarly, each internal processor 
sends $t$ copies of the message that it received in the previous superstep 
to its own children.  The algorithm completes within at most $\ceil{n/m}+h-1$ 
supersteps. Each superstep requires $gtm$ communication time, and the lemma 
follows.  
\end{proof}

\begin{lemma}
\label{Prefix}
There exists a BSP algorithm for computing $n$ independent parallel-prefix 
operations that requires time at most
$C^{n}_{\mathit{ppf}}(p) = 2(\ceil{n/\ceil{n/h}}+h-1) \mx{L, t\ceil{n/h}}$
and $M^{n}_{\mathit{ppf}}(p) = 2(\ceil{n/\ceil{n/h}}+h-1) \mx{L, 
g2t\ceil{n/h}}$,  for computation and communication respectively, for any 
integer $2 \leq t  \leq  p$, where $h= \ceil{\log_{t}{p}}$, for a total 
time of  $T^{n}_{\mathit{ppf}}(p) = C^{n}_{\mathit{ppf}}(p) + 
M^{n}_{\mathit{ppf}}(p)$.
\end{lemma}

\begin{proof}
The underlying algorithm consists of two passes of the algorithm implied
by Lemma~\ref{Broadcast}, on a pipelined $t$-ary tree, consisting of $p$ 
leaf nodes    and at most $p$ internal nodes, for some appropriate $t$.  The 
depth of the tree is $h=\ceil{\log_{t}{p}}$.  The lemma follows by way of
arguments similar to that of Lemma~\ref{Broadcast}.
\end{proof}

For $n=1$,  $T_{\mathit{ppf}}(p) $ (respectively $C_{\mathit{ppf}}(p)$, 
$M_{\mathit{ppf}}(p)$) can be written for $T^{1}_{\mathit{ppf}}(p)$ (respectively 
$C^{1}_{\mathit{ppf}}(p)$, $M^{1}_{\mathit{ppf}}(p)$).  The same notational 
convention applies to all pipelined operations. 

\section{The Algorithms}
\label{UAlg}
\label{s5}

\subsection{BSP Deterministic Sorting Algorithm in \cite{GS96a,GS99a}}
\label{detalg}

The deterministic algorithm of the implementations is based on a 
non-iterative variant of the  sorting algorithm of \cite{GS96a,GS96f,GS99a} 
which has been shown to be one-optimal for a satisfactory range of the 
BSP parameters that includes most currently built parallel machines;  
for a wider range of these parameters the algorithm is $c$-optimal, 
where $c \geq 1$ is a small constant.  The algorithm is regular-sampling
based (\cite{SH}) but extends regular sampling to regular oversampling
and utilizes an efficient
partitioning scheme that splits -- almost evenly and independently of
the input distribution -- an arbitrary number of sorted sequences
among the processors.  In Section \ref{duplicate} this base algorithm is
augmented to handle transparently and in optimal asymptotic efficiency
duplicate keys. In our approach duplicate handling does not require
doubling of communication and/or computation time that other
approaches seem to require \cite{jaja1,jaja2,jaja3}.
The base algorithm consists of the following phases:
\begin{elist}
\item {\em Local Sorting.} Each processor, in parallel with all the
  other processors, sorts its local input sequence of size ${n/p}$ that
  resides in its local memory.
\item {\em Partitioning.} Processors cooperate to evenly split the
  sorted sequences among the processors.
\item {\em Merging.} Each processor, in parallel with all the other
  processors, merges a small number of subsequences of total size
  $(1+o(1)) (n/p)$.
\end{elist}
The iterative version of the algorithm performs $m$ iterations of
phases 2 and 3.  In each iteration the data set is partitioned into
$k$ buckets of approximately equal size. In the following iteration a
similar process of sub-partitioning each of these $k$ buckets into $k$
further buckets is performed, and so on.  The maximum number of keys
per processor in any of the $m$ iterations can be shown to be $(1 +
O(1/\omega_{n})) (n/p)$. By employing multi-way merging \cite{Knuth73} 
the desired sorted sequence can be obtained.

In Figure~\ref{SortD}, function {\sc Sort\_Det\_BSP}($X$), 
implements the sorting operation for $m=1$. 
$X$ denotes the input sequence, $n$ the size of $X$, 
$p$ the  number of processors, and $s$ is a user defined parameter 
(oversampling factor) that  inversely relates to the maximum possible 
imbalance of the sequences formed in step~(13) of the algorithm. 
For any sequence $X$, the subsequence of $X$ 
residing in processor $k$ is denoted $X^{\langle k \rangle}$.

The implemented version {\sc Sort\_Det\_BSP} as a special
case of the more general algorithm reported in  \cite{GS96a,GS96f,GS99a} 
can also be viewed as an adaptation of the regular-sampling algorithm 
of \cite{SH} on the BSP model. It differs, however, from the algorithm
in \cite{SH} in that sample-sorting is performed in parallel and
deterministic oversampling is used with the purpose of achieving finer
load balancing in communication. The choices
in deterministic oversampling that also affect load-balancing in the
subsequent key routing step are based on the quantifiable results 
obtained from the detailed analysis of \cite{GS99a}; although oversampling 
has been used previously in randomized sorting, its usefulness in 
deterministic sorting was not explored and quantified in full.

\begin{figure}[htb]
\centering
\fbox{
{\small
\begin{minipage}{\textwidth}
\begin{tabbing}
mmi \= mi \= mi \= mi \= mi \= mi \= mi \kill
{\bf begin} {\sc Sort\_Det\_BSP} ($X$) \\
1. \> {\bf let} $r =  \ceil{\omega_n} $ ; \\
2. \> {\bf for} each processor ${\langle k \rangle}$,
         $k \in \{0,\ldots,p-1\}$, in parallel \\
3. \> \> {\bf do} \> {\sc Sort\_Seq}($X^{\langle k \rangle}$) ; \\
4. \> \> \> {\bf form  } locally a sample $Y^{\langle k \rangle} = \langle
              y_{1},\ldots, y_{r p-1 } \rangle$ of $rp-1$ evenly spaced \\
   \> \> \> \> keys that partition $X^{\langle k \rangle}$ into
              $s=rp$ evenly sized segments \\
   \> \> \> \> and append the maximum of $X^{\langle k \rangle}$ 
               to this sequence ; \\
5. \> {\bf let} $\bar{Y} = \mbox{{\sc Bitonic\_Sort}($Y$)}$ ; \\
6. \> {\bf form} ${S} = \langle s_{s},\ldots, s_{(p - 1)s} \rangle$ from 
       $\bar{Y}$ that consistsof $p - 1$ evenly spaced splitters  \\
   \> \> that partition  $\bar{Y}$ into $p$ evenly sized segments ; \\
7. \> {\sc Broadcast} $({S})$; \\
8. \>  {\bf for} each processor ${\langle k \rangle}$,
         $k \in \{0,\ldots,p-1\}$, in parallel \\
9. \> \> {\bf do} \> {\bf partition} $X^{\langle k \rangle}$ into $p$
              subsequences $X^{\langle k \rangle}_{0}, \ldots,
              X^{\langle k \rangle}_{p-1}$ as induced by the \\
   \> \> \> \> $p-1$ splitters of ${S}$ ; \\
10.\> \> \> {\bf for} all $i \in \{0,\ldots,p-1\}$ \\
11.\> \> \> \> {\bf do} \> {\bf communicate} subsequence $X^{\langle k
              \rangle}_{i}$ to $Z^{\langle i \rangle}_{k}$ ; \\
12.\> \> \> {\bf let} $\bar{X}^{\langle k \rangle} =
              \mbox{{\sc Merge\_Seq}($\bigcup_{i}Z^{\langle k
              \rangle}_{i}$)}$ ; \\
13.\> \> \> {\bf return} $\bar{X}^{\langle k \rangle}$ ; \\
{\bf end} {\sc Sort\_Det\_BSP}
\end{tabbing}
\end{minipage}
}
}
\caption{{\small Procedure {\sc Sort\_Det\_BSP}.}}
\label{SortD}
\end{figure}

The proposition and proof below simplify the results shown in a more
general context in  \cite{GS96a,GS96f,GS99a}.

\begin{proposition}
\label{detsorting}
For any $n$ and $p \leq n$, and any function 
$\omega_{n}$ of $n$ such that $\omega_{n}=\Omega(1)$, 
$\omega_n = O(\lg{n})$
and $p^2 \omega_n^2 \leq n/\lg{n}$, and for
$n_{\mathit{max}} = (1 + 1/\ceil{\omega_{n}}) n/p+ \ceil{\omega_n} p $, 
algorithm {\sc Sort\_Det\_BSP}
requires time, 
$(n/p) \lg{(n/p)}+ n_{\mathit{max}}\lg{p}+
O(p+ \omega_n p \lg^2{p})$ for computation and 
$g n_{\mathit{max}}+L\lg^{2}{p}/2 + O(L+g\omega_n p \lg^2{p})$ for 
communication. 
%Moreover, 
%$\pi = 1 + \lg{p}/(\ceil{\omega_n} \lg{n}) + 
%O( 1/(\omega_n \lg{n})+ \lg^2{p}/(\omega_n \lg^2{n}))$, and
%for $L \leq 2n/ ( p \lg^2{p} )$,
%$\mu =(1+1/ \ceil{\omega_n} )g/\lg{n}+L p \lg^2{p} /(2n \lg{n}) +
%O(g \lg^2{p} / (\omega_n \lg^2{n})+ 1/(\lg{n}\lg^2{p}))$ as well.
\end{proposition}

\begin{corollary}
For $n$, $p$ and $\omega_n$ as in Proposition \ref{detsorting},
algorithm {\sc Sort\_Det\_BSP} is such that
$\pi = 1 + \lg{p}/(\ceil{\omega_n} \lg{n}) + 
O( 1/(\omega_n \lg{n})+ \lg^2{p}/(\omega_n \lg^2{n}))$, and
for $L \leq 2n/ ( p \lg^2{p} )$,
$\mu =(1+1/ \ceil{\omega_n} )g/\lg{n}+L p \lg^2{p} /(2n \lg{n}) +
O(g \lg^2{p} / (\omega_n \lg^2{n})+ 1/(\lg{n}\lg^2{p}))$ as well.
\end{corollary}

\begin{proof}{}
The input is assumed to be evenly but otherwise arbitrarily distributed 
among the $p$ processors before the beginning of the execution of the 
algorithm.  Moreover, the keys are distinct since in an extreme case,
we can always make them so by, for example, appending to them the code 
for their memory location. We later explain how we handle duplicate keys
without doubling (in the worst case) the number of comparisons performed.
Parameter $\omega_n$ determines the desired upper bound in processor
key imbalance during the key routing operation. The term
$1+ 1/ \ceil{\omega_n}$ is also referred to as bucket expansion in sample-sort
based randomized sorting algorithms (\cite{Blelloch91}).

In step 3, each processor sorts the keys in its possession.
As each processor holds at most $\ceil{n/p}$ keys, this step
requires  time $\ceil{n/p} \lg{\ceil{n/p}}$. 
Algorithm {\sc Sort\_Seq} is any sequential sorting algorithm of such
performance.

Subsequently, each processor selects locally $\ceil{\omega_{n}} p - 1=rp-1$ 
evenly spaced sample keys, that partition its input into 
$\ceil{\omega_{n}} p$ evenly sized segments.  Additionally, each 
processor appends to this sorted list the largest key in its input.  
Let $s = \ceil{\omega_{n}} p=rp$ be the size of the so identified list.
Therefore step 4 requires time $O(s)$.

The $p$ sorted lists, each consisting of $s$ sample keys, are merged; 
let sequence $\langle s_{1}, s_{2}, \ldots, s_{ps}\rangle$ be the result of 
the merge operation.  By assumption, the sequence is evenly distributed among 
the $p$ processors, i.e., subsequence 
$\langle s_{is+1}, \ldots, s_{(i+1)s}\rangle$, $0 \leq i \leq p-1$, 
resides in the local memory of the $i$-th processor.
The cost of step 5 is that of parallel sorting by one of
Batcher's methods \cite{Batcher68}, appropriately modified
\cite{Knuth73} to handle sorted sequences of size $s$.  
The computation and communication time required for this stage is,
respectively, $2s (\lg^{2}{p}+\lg{p})/2$ and 
$(\lg^{2}{p}+\lg{p})(L+g s)/2$.

In step 6,  a set of evenly spaced splitters is formed from the sorted
sample. A broadcast operation is initiated in step 7, where splitter $s_{is}$, 
$1 \leq i < p$, along with its index in the sequence of
sample keys is sent to all processors.
Lines 6 and 7 require time $O(1)$ and $\mx{L, g O(p)}+ 
M_{\mathit{brd}}^{p-1}(p)$ for computation and communication respectively.

In step 9, each processor decides the position of every key it holds 
with respect to the $p-1$ splitters it received in step 7, by
way of sequential merging the splitters with the input
keys in $p-1+n/p$ time or alternately by performing a binary search
of the splitters into the sorted keys in time $p \lg{(n/p)}$, 
and subsequently counts the number of keys that fall into each of the 
$p$ so identified buckets induced by the $p-1$ splitters. 
Subsequently, $p$ independent parallel prefix
operations are initiated (one for each subsequence) to determine how to
split the keys of each bucket as evenly as possible among the
processors using the information collected in the merging operation.
The $p$ disjoint parallel prefix
operations in step 9 are realized by employing the algorithm of
Lemma~\ref{Prefix} thus resulting in a time bound of
$p\lg{(n/p)}+T_{\mathit{ppf}}^{p}(p)$ for step 9.

In step 11,  each processor uses the information collected  by the
parallel prefix operation to perform the routing in such a way that
the initial ordering of the  keys is preserved 
(i.e. keys received from processor $i$ are stored before those received 
from $j$, $i<j$, and also the ordering within $i$ and $j$ is also 
preserved). Step 11 takes time $\mx{L, g \mmul n_{\mathit{max}}}$. 
  
In  step 12, each processor merges the at most $p$ sorted subsequences that 
it received in step 11.  
When this step is executed, each processor,
by way of Lemma~\ref{balance} to be shown,
possesses at most $p=\mn{p, n_{\mathit{max}}}$ sorted 
sequences for a total of at most $n_{\mathit{max}}$ keys,  
where 
$n_{\mathit{max}} =(1+1/\ceil{\omega_{n}})  (n/p) +\ceil{\omega_n} p$.
The cost of this stage is that of sequential multi-way merging
$n_{\mathit{max}}$ keys by some deterministic algorithm \cite{Knuth73}, 
which is $n_{\mathit{max}} \lg{p} $, as $\omega_n^2 p = O(n/p)$.

\bigskip
%{\bf Remark.} In case the (sorted) output sequence needs to be balanced
%among the processors at the end of the computation, a balancing
%operation is performed in time $O(n_{\mathit{max}}) +
%C_{\mathit{ppf}}^{1}(p)$ and $\mx{L, g \mmul O(n_{\mathit{max}})} +
%M_{\mathit{ppf}}^{1}(p)$ for computation and communication
%respectively.  We note that this latter term does not affect the
%asymptotic complexity of the algorithm.
%
\noindent
Summing up all the terms  for computation and communication and noting
the conditions on $L$ and $g$ and assigning a cost of $n\lg{n}$ to the
best sequential algorithm for sorting the result follows.
\end{proof}

\medskip
It remains to prove that at the completion of  step 9 the input keys are 
partitioned into (almost) evenly sized subsequences. 
The main result is summarized in the following lemma.

\begin{lemma}
\label{balance}
The maximum number of keys $n_{\mathit{max}}$ per processor in
{\sc Sort\_Det\_BSP}
is $(1+1/\ceil{\omega_{n}})  (n/p) +\ceil{\omega_n} p$, for any 
$\omega_{n}$ such that
$\omega_{n} = \Omega(1)$ and $\omega_n = O(\lg{n})$, provided that 
$\omega_n^2 p = O(n/p)$ is also satisfied.
\end{lemma}

\begin{proof}{}
Although it is not explicitly mentioned in the description of
algorithm {\sc Sort\_Det\_BSP} we assume that we initially pad the
input so that each processor owns exactly $\ceil{n/p}$ keys. At most
one key is added to each processor (the maximum key can be such a
choice).  Before performing the sample selection operation, we also pad 
the input so that afterwards, all segments have the same number of keys 
that is, $x =\ceil{\ceil{n/p}/s}$. The padding operation requires time 
at most $O(s)$, which is within the lower order terms 
of the analysis of Proposition~\ref{detsorting}, and therefore, does not 
affect the asymptotic complexity of the algorithm.  We note that padding
operations introduce duplicate keys; a discussion of duplicate handling
follows this proof.

Consider an arbitrary splitter $s_{is}$, where $1 \leq i < p$. 
There are at least $i s x$ keys which are not larger than $s_{is}$,
since there are $i s$ segments
each of size $x$ whose keys are not larger than $s_{is}$.  Likewise, 
there are at least $(ps - i s - p + 1) x $ keys which are not smaller 
than $s_{is}$, since there are $p s - i s - p + 1$ segments each
of size $x$ whose keys are not smaller than $s_{is}$.  Thus, by noting 
that the total number of keys has been increased (by way of padding 
operations) from $n$ to $p s x$, the number of keys
$b_{i}$ that are smaller than $s_{is}$ is bounded as follows.
\[
 i s x \leq b_{i} \leq p s x - \lf( ps- i s - p + 1 \rg) \mmul x.
\]
A similar bound can be obtained for $b_{i+1}$.  Substituting
$s= \ceil{\omega_{n}} p$ we therefore
conclude the following.
\[ 
  b_{i+1} - b_{i  } \leq  sx +px-x \leq sx+px =
    \Ceil{\omega_{n}} p x + p x.
\]
The difference $n_{i} = b_{i+1} - b_{ i }$ is independent
of $i$ and gives the maximum number of keys per split sequence. 
Considering that $x \leq (n + p s)/(ps)$ and substituting $s=
\ceil{\omega_{n}} p $, the following bound is
derived.
\[
  n_{\mathit{max}} =
  \lf(1+\frac{1}{\Ceil{\omega_{n}} }\rg)
      \frac{n+ps}{p}.
\]
\noindent 
By substituting in the  nominator of the previous expression
$s= \ceil{\omega_{n}} p$,
we conclude that the maximum number of keys $n_{\mathit{max}}$ per
processor of function {\sc Sort\_Det\_BSP} is bounded above as follows.
\[
  n_{\mathit{max}} = \lf(1 + \frac{1}{\ceil{\omega_{n}}}\rg) 
    \frac{n}{p} + \ceil{\omega_{n}} p .
\]
\noindent 
The lemma follows.
\end{proof}

The particular algorithm as depicted in Figure  \ref{SortD} corresponds to
the simplest case of the deterministic algorithm in \cite{GS99a} where an
analysis for all possible values of $p$ is presented.  In the general
algorithm, the number of communication rounds is 
$\lg{n}/\lg{(n/p)}$ for any $p=n^{1-t}$, with $t$ constant, i.e. it is constant. 
If, however, $t$ is not constant, then the number of rounds becomes
$O( \lg{n} / \lg{(n/p)})$ still matching in all cases the lower bound 
of \cite{Goodrich96}. We note that in the general case, the implementation of
parallel prefix and broadcasting operations depends on $p$, $L$, $g$ and
the amount of information processed by these operations.
This highlights a difference between {\em architecture independent parallel} 
algorithm design and {\em classic parallel } algorithm design. For a 
given choice of problem size $n$, and machine 
i.e. for a given tuple $(n,p,L,g)$, the resulting algorithm may differ from 
that chosen for another tuple $(n^\prime ,p^\prime ,L^\prime ,g^\prime )$. 
In the case of sorting for example, one algorithm may implement
sample sorting sequentially while another one in parallel, one algorithm
may implement a parallel prefix or broadcasting operation using a PRAM
approach in $\lg{p}$ supersteps while another algorithm may
perform the same operations in constant number of supersteps as in
Lemma \ref{Broadcast} or \ref{Prefix}. 

\subsubsection{Duplicate-key Handling}
\label{duplicate}
\label{sd}

Algorithm {\sc Sort\_Det\_BSP}, as described, does not handle 
duplicate keys. A naive way
to handle duplicate keys is by making the keys distinct. This could 
be achieved
by attaching to each key the address of the memory location it is
stored in.  For data types whose bit or byte length is comparable to
the length of the address describing them,
such a transformation leads -- in most cases -- to a doubling of the 
overall number of comparisons performed and the communication time
in the worst case.  For
more complex data types such as strings of characters the  extra cost
may be negligible.

An alternative way to handle duplicate keys is the following
one that was also used in the implementations and handles duplicate keys
in a transparent way that provides asymptotic optimal efficiency and
tags only a small fraction of the keys. This seems to be an improvement
over other approaches \cite{jaja1,jaja2,jaja3} that require a doubling
of communication time.
Procedure {\sc Sort\_Seq} is implemented
by means of a stable sequential sorting algorithm. 
Two tags for each input key are already implicitly available by default,
and no extra memory is required to access them.
These are the processor identifier that stores a particular input key
and the index of the key in the local array that stores it. No additional
space is required for the maintenance of this tagging.
In our duplicate-key handling method such tags are only used for sample
and splitter-related activity.

As sample sorting is a global operation,
for sample sorting and splitter distribution only we augment
every sample key 
into a record that includes this additional tag information
(array index and processor storing the key).
As the additional tagging information affects the sample only,
and the sample is $o(1)$ of the input keys, the memory overhead
incurred is small, as is the computational overhead.
The attached tag information is used in  step  4  to form the
sample, in step 5 for sample sorting, in steps 6  and 7 
for splitter selection and broadcasting, and finally in step 9 as all
these steps require distinct keys to achieve stability and
load-balance. In step 9 in particular, a binary search operation 
of a splitter key into the locally sorted keys involves  first
a  comparison of the two keys. If the keys are equal the 
result of the comparison is resolved by comparing the readily 
(and implicitly) available identifier of the processor that holds 
the local key to the processor storing the comparing splitter 
(available through the tagging in step 4, and the broadcasting of step 7). 
If the two processor identifiers are equal,
then the result of the comparison is determined by comparing the
indices of the position in the local array that stores the
local key and the splitter being compared.

In addition, the merging operation must also be performed in a 
stable manner, that is if, the keys at the head of two sorted sequences
are equal the one received from processor $i$ is appears before the one 
received from processor $j$, $i<j$ in the output sequence of the operation.  

The computation and communication overhead of duplicate handling that 
is described by this method 
is within the lower order terms of the analysis
and therefore, the optimality claims still hold unchanged.
The results on key imbalance still hold as well.
This same duplicate handling method is also used in the
implementation of algorithm {\sc Sort\_IRan\_BSP}.

\subsection{BSP Randomized Sorting Algorithm in \cite{GV94}}

Randomized BSP sorting was introduced in \cite{GV94}.
An architecture independent analysis of the algorithm in \cite{GV94} in
terms of problem size $n$ and $p, L$ and $g$ shows
that it is one-optimal for most cases of interest.
The algorithm derives from quicksort, the ideas of 
\cite{Frazer70,Reischuk85} and the technique of oversampling 
\cite{Reif87}. It achieves the claimed efficiency as follows.

(1). The algorithm satisfies the requirement of one-optimality by using
oversampling. The oversampling factor is in general $\omega (\lg{n})$
and a practical choice that is being used in the experiments is
$\Theta (\lg^2{n})$.

(2). As the size of the sample is smaller than input size $n$
and parallel sampling is used, sample sorting can be performed
either sequentially (by sending the sample to a single processor and sorting
locally) or, as noted in \cite{GV94}, recursively, or by employing a 
non-optimal but in practice faster 
parallel algorithm such as  Batcher's bitonic or odd-even merge sort
as any of the latter two is  simpler to implement and incurs low 
overhead costs 
for small problem instances.

(3). At every recursive call of classic quicksort an input sequence 
is split into two subsequences; a naive parallel implementation of 
quick sort would require in the best case $\lg{n}$ communication rounds, 
one for each recursive call, each round requiring $O(gn/p)$ time for 
communication. At the end of the $\lg{p}$-th round, 
the input is split into $p$ segments which is equal to the number 
of available processors. In a well-designed parallel implementation
of quick sort, communication is only incurred in the first $\lg{p}$ 
rounds for a total of $O(g n \lg{p}/p)$.
Total computation time is, however, $O(n \lg{n} /p)$.  
If $p$ polynomially related to $n$, i.e. $p = n^{1-t}$, where 
$0< t < 1$ is a constant, such an implementation is still inefficient as  
communication time is of the order of computation time 
even for $g=O(1)$.

(4). Based on ideas of \cite{Frazer70} it then makes sense to split the input
sequence into $k$ sets, where $k > 2$, by choosing $k-1$ splitters at
every phase. This way the number of communication rounds 
is reduced to  $\lg{p}/\lg{k}$. For $p=n^{1-t}$ as before, by choosing $k$ 
so that $k=p^{t}$, the number of communication rounds is then 
$\lg{n}/\lg{(n/p)}$. For constant $t$, this is constant
and therefore, communication time is smaller ($O(g n/p)$) than before. 
This observation is used in the randomized algorithm of 
\cite{GV94} that describes the various cases of interest:
(a) $p \leq \sqrt{n}$ that is of practical interest, (b) $p= n^{1-t}$, where
$t$ is constant and the algorithm is still interesting in terms of its
practical implications and (c) for all other cases, $p$ can grow as large
as $p=O(n / \lg^{1+a}{n})$, for any constant $a>0$ and 
one-optimality can still be maintained. 
For most practical applications the number $m$ of communication
rounds is one (case (a)) or in some extreme cases at most 2 
(case (b)). Which of the algorithms in (a) or (b) applies depends
on the values of $n$, $p$, $L$ and $g$.

Although partitioning and oversampling in the context of sorting
are well established techniques, the analysis in \cite{GV94} summarized in 
Claim \ref{Sampling1} below allows one to prove the one-optimality of
the sorting algorithm by quantifying precisely the key imbalance among
the processors.
Let $X = \langle x_{1}, x_{2}, \ldots , x_{N} \rangle$ be an ordered 
sequence of keys indexed such that $x_{i} < x_{i+1}$, for all 
$1 \leq i  \leq N-1$. The implicit assumption is that keys are unique. 
Let
$Y = \{y_{1}, y_{2}, \ldots, y_{ks-1}\}$ be a randomly chosen subset of $ks-1
\leq N$ keys of $X$ also indexed such that $y_{i} < y_{i+1}$, for all
$1 \leq i  \leq ks-2$,  for some positive integers $k$ and $s$. Having
randomly selected set $Y$, a partitioning of  $X - Y$ into $k$ 
subsets, $X_{0}, X_{1}, \ldots, X_{k-1}$ takes place.
The following result shown in \cite{GV94} 
is independent of the distribution of the input keys.

\begin{claim}
\label{Sampling1}
Let $k \geq 2$, $s \geq 1$, $ks < N/2$, $n \geq 1$,  $0 < \varepsilon < 1$,
$\rho >0$, and

\[
  s \geq \frac{1+\varepsilon}{\varepsilon^{2}} \lf( 2 \rho \log{n} + 
         \log{(2 \pi k^{2}(ks-1)e^{1/(3(ks-1))})} \rg).
\]
\noindent
Then the probability that any one of the $X_{i}$, for all $i$, 
$0 \leq i \leq k-1$, is of size more than 
$\ceil{(1+\varepsilon)(N-k+1)/k}$ is at most $n^{-\rho}$.
\end{claim}

Algorithm {\sc Sort\_Ran\_BSP} describes case (4)(a) of \cite{GV94},
and this special case is widely referred to as sample-sort in the
literature.
$X$ denotes the input key sequence, $n$ the size of $X$, 
$p$ the number of processors, and $s$ is a user defined parameter 
(oversampling factor) that inversely relates to the maximum key imbalance 
of the split sequences formed in step~(13) of the algorithm.

\begin{figure}[htb]
\centering
\fbox{
{\small
\begin{minipage}{\textwidth}
\begin{tabbing}
mmi \= mi \= mi \= mi \= mi \= mi \= mi \kill
{\bf begin} {\sc Sort\_Ran\_BSP} ($X$) \\
1. \> {\bf let} $s = 2\omega_n^2 \lg{n}$ ; \\
2. \> {\bf select} uniformly at random a sample $Y = \langle y_{1}, \ldots,
              y_{sp - 1} \rangle$ of $sp - 1$ keys ; \\
3. \> {\bf communicate} $Y$ to processor ${\langle 0 \rangle}$ ; \\
4  \> if processor ${\langle 0 \rangle}$ {\bf then} \\
5. \> \> {\bf do} \> {\bf let} $\bar{Y} = \mbox{{\sc Sort\_Seq}($Y$)}$ ; \\
6. \> \> \> {\bf form} locally a set  $S = \langle s_{1},\ldots, s_{p - 1}
              \rangle$ of $p - 1$ evenly spaced splitters \\
   \> \> \> \> that partition $\bar{Y}$ into $p$ evenly sized segments ; \\
7. \> \> \> {\sc Broadcast}($S$) ; \\
8. \>  {\bf for} each processor ${\langle k \rangle}$,
         $k \in \{0,\ldots,p-1\}$, in parallel \\
9. \> \> {\bf do} \> {\bf partition} $X^{\langle k \rangle}$ into $p$
              subsets $X^{\langle k \rangle}_{0}, \ldots,
              X^{\langle k \rangle}_{p-1}$ as induced by the $p-1$\\
   \> \> \> \> splitters of $S$ ; \\
10.\> \> \> {\bf for} all $i \in \{0,\ldots,p-1\}$ \\
11.\> \> \> \> {\bf do} \> {\bf communicate} subset $X^{\langle k
              \rangle}_{i}$ to $Z^{\langle i \rangle}_{k}$ ; \\
12.\> \> \> {\bf let} $\bar{X}^{\langle k \rangle} =
              \mbox{{\sc Sort\_Seq}($\bigcup_{i}Z^{\langle k
              \rangle}_{i}$)}$ ; \\
13.\> \> \> {\bf return} $\bar{X}^{\langle k \rangle}$ ; \\
{\bf end} {\sc Sort\_Ran\_BSP}
\end{tabbing}
\end{minipage}
}
}
\caption{{\small Procedure {\sc Sort\_Ran\_BSP}.}}
%\caption{{\small Procedure {\sc Sort\_Ran\_BSP}, $n$ - problem size,
%$p$ - number of processors, $s$ - user defined parameter that is
%inversely related to the maximum possible imbalance.}}
\label{SortR}
\end{figure}

\begin{proposition}
\label{ransorting}
For any $\omega_n , n , p$ such that $2p \omega_n^2 \lg{p} <n/2$, $p^2 \leq n$,
and $L=o(n/(p \lg^2{p}))$, algorithm {\sc Sort\_Ran\_BSP} has
$\pi = 1+ 1/\omega_n + 1/\lg{n}+ 2p^2 \omega_n^2 \lg{p}/n+
o(1/\lg{n})$ and $\mu = O(g p^2 \omega_n^2 /n)+O(g/\lg{n})+o(g/\lg{n})$.
\end{proposition}

\begin{proof}
For the sake of completeness we use    elements of the proof 
of \cite{GV94} to illustrate the performance of the algorithm for the
specific choice of splitters $k=p-1$. 

Step 2 of the algorithm requires parallel time $O(L+gsp)$ per processor.
As shown in \cite{GV94}, processors select
uniformly at random processor identifiers $0 \ldots p-1$
and send  these identifiers to the identified processors. Chernoff
bound techniques can be used to show that each processor sends
or receives $O(s)$ identifiers for communication time $O(L+gs)$.
Processors then select uniformly at random and without replacement a number 
of keys equal to the number of identifiers received previously, a step that
takes $O(s)$ parallel time. Since all sample keys are then communicated 
to processor 0 this step would require $O(L+gsp)$ time.
Step 5 requires time $sp \lg{(sp)}$ and
splitter selection in step 6 takes $O(p)$ time.

%\medskip
%{\em Note: } The communication in step 3 can be avoided if parallel 
%bitonic-sort is used instead of sequential sorting. 
%In this case sample sorting
%would require time $O(s \lg^2{p} + L \lg^2{p} + g s \lg^2{p})$ for
%computation and communication (steps 3-5) instead of
%$sp \lg{(sp)}+ O(L+gsp)$ for sorting and routing in steps 3-5. 
%Step 6, however, will take $O(1)+O(gp)$ time, as each processor 
%selects at most one splitter and subsequently,
%the splitters are communicated to processor 0 that will initiate the
%broadcasting of step 7.
%\medskip

The broadcasting in step 7 takes $O(L+gp)$ time as it will take place 
in one superstep. Each processor $k$ then performs a binary search of its
set $X^{\langle k \rangle}$ of $n/p$ input keys  onto the $p-1$ splitters to
determine sets
$X^{\langle k \rangle}_0  \ldots X^{\langle k \rangle}_{p-1}$. 
This requires at most  $(n/p)(\lg{p}+1)$ comparisons. 
%(Alternately, 
%a merging operation of the splitters and the input keys would 
%also allow the determination  of 
%$X^{\langle k \rangle}_0  \ldots X^{\langle k \rangle}_{p-1}$; the 
%running time would be $n/p+p$ instead).
By claim \ref{Sampling1}, with probability $1-o(1)$, each of 
$\bar{X}^{\langle k \rangle}$  is of size at most $(1+1/\omega_n ) n/p$.
Therefore, in  step 11, each processor sends
$n/p$ and receives at most $(1+1/\omega_n ) n/p$ keys for communication time
$O(L+g(1+1/\omega_n ) n/p)$ and computation time in step 12 of
$(1+1/\omega_n ) n/p \lg{(n/p)} +o(n/p)$ for local sorting, 
as $\ln{(1+x)} \leq x$, $x<1$.

The total parallel computation time of {\sc Sort\_Ran\_BSP} is 
$(1+1/\omega_n ) n \lg{n}/p +(n/p) +sp\lg{(sp)}+o(n/p)+O(L+p)$ and 
communication time is $O(g n /p + gps +L)$. 
If we compare the parallel
algorithm  to the best sequential algorithm that requires 
$n \lg{n}$ comparisons
and noting that $L=o(n /(p\lg^2{p}))$, the claimed bounds on $\pi$ and
$\mu$ are derived.
%
%{\em Note:}
%Had we used however parallel sample sorting as outlined in the previous
%note, parallel  runtime would have been
%$(1+1/\omega_n ) n \lg{n}/p +(n/p)+O(L \lg^2{p})+O(p+s\lg^2{p})$ and 
%$O(g n /p + gp+L\lg^2{p}+gs \lg^2{p})$.
%Then, for $L=o(n/(p \lg^2{p}))$ and $p^2 =o(n)$, we get
%$\pi = 1 +1/\omega_n + 1/\lg{n}+o(1/\lg{n})$ and 
%$\mu = O( g/\lg{n})+ o(g/\lg{n})$.
%
\end{proof}

\medskip
{\sc Sort\_Ran\_BSP} maintains an oversampling factor 
$2 \omega_n^2 \lg{n}$ that is $\Omega ( \lg{n})$.
\medskip

An implementation of {\sc Sort\_Ran\_BSP} of Figure~\ref{SortR} 
is straightforward except perhaps that of step 9. 
In step 9, set $X^{\langle k \rangle}$ is split into $p$ sets such 
that the keys of the $i$-th set are routed to processor $i$; the sets 
are determined by a binary search operation of each key into the set of
$p-1$ splitters. 
The formation  of the $p$ sets in step 9 is equivalent
to an integer sort operation with key the result of the binary search 
operation (i.e. destination processor).
Such set formation operation is thus of linear time $Dn/p$; constant $D$ is
significant however since it includes the cost of copying keys  in
memory so that keys with the same destination processor are stored
in contiguous memory locations.
For relatively small values of $n$, constant $D$ has the same 
growth as $\lg{n}$, and therefore, asymptotic claims may not be valid. 
In addition it seems that the sorting operation of
step 5  could be done in parallel rather than sequentially.

These observations were taken into consideration in the design
of the randomized BSP sorting algorithm of the implementations so that 
its performance is fully optimized.  The resulting algorithm
is {\sc Sort\_IRan\_BSP}. It is this algorithm that was implemented
rather than  sample-sort {\sc Sort\_IRan\_BSP}.
Note that {\sc Sort\_IRan\_BSP} differs from most traditional
sample-sort randomized parallel sorting algorithms in that it
follows the pattern of local sorting, sample and splitter selection,
key routing and local merging rather than the traditional one of
sample and splitter selection, key routing and local sorting that
identifies traditional sample-sort based randomized parallel sorting algorithms.
In particular, in the design and implementation of 
{\sc Sort\_IRan\_BSP} we  addressed 
the first limitation related to step 9 of {\sc Sort\_Ran\_BSP}.
To this end we employed ideas from the deterministic algorithm 
\cite{GS96a,GS96f} described in {\sc Sort\_Det\_BSP}. 
In particular, we sorted the input keys before
realizing the communication in step~(11) and before we performed the sampling
operation. This requires the replacement of the sequential sorting algorithm 
in step~(12) of {\sc Sort\_Ran\_BSP} by 
a multi-way merging algorithm as the received
sequences from at most $p-1$ processors are already sorted.  
A binary search in step~(9) is not required any more as we can merge the
local sorted keys with the sorted splitters  or perform a binary
search of the splitters into the keys an operation that
allows for coarse-grained communication in step~(11). The resulting
algorithm looks similar to {\sc Sort\_Det\_BSP} and thus, 
attains performance at least comparable to that of the corresponding 
deterministic algorithm.  Random sampling, however, gives the
programmer more freedom to determine processor imbalance ($\omega_n$ in
the deterministic case is $O(\lg{n})$; there is no such limitation in the
randomized case).
For handling duplicate keys we employed the
method of the deterministic algorithm implementation described in
Section \ref{duplicate} which effectively
tags few keys only (sample keys) and thus uses asymptotically as much memory
as the non-duplicate handling variant. Algorithm
{\sc Sort\_IRan\_BSP} is outlined in Figure \ref{SortRI}. 

\begin{figure}[htb]
\centering
\fbox{
{\small
\begin{minipage}{\textwidth}
\begin{tabbing}
mmi \= mi \= mi \= mi \= mi \= mi \= mi \kill
{\bf begin} {\sc Sort\_IRan\_BSP} ($X$) \\
1. \> {\bf let} $s = 2 \omega_n^2 \lg{n}$ ; \\
2. \> {\bf for} each processor ${\langle k \rangle}$,
         $k \in \{0,\ldots,p-1\}$, in parallel \\
3. \> \> {\bf do} \> {\sc Sort\_Seq}($X^{\langle k \rangle}$) ; \\
4. \> \> \> {\bf select} uniformly at random a sample $Y = \langle y_{1}, \ldots,
              y_{sp - 1} \rangle$ of $sp - 1$ keys ; \\
5. \> {\bf let } $\bar{Y}= Bitonic\_Sort (Y)$. \\
6. \> {\bf for} each processor ${\langle k \rangle}$,
         $k \in \{0,\ldots,p-1\}$, in parallel \\
7. \> \> \>{\bf form} set  $S = \langle s_{1},\ldots, s_{p - 1} \rangle$ 
         of $p - 1$ evenly spaced splitters \\
   \> \> \> \> that partition $\bar{Y}$ into $p$ evenly sized segments ; \\
8. \> \> \>{\bf communicate} the splitters to processor 0. \\
9. \> {\sc Broadcast}($S$) ; \\
10. \>  {\bf for} each processor ${\langle k \rangle}$,
         $k \in \{0,\ldots,p-1\}$, in parallel \\
11. \> \> {\bf do} \> {\bf partition} $X^{\langle k \rangle}$ into $p$
              subsequences $X^{\langle k \rangle}_{0}, \ldots,
              X^{\langle k \rangle}_{p-1}$ as induced by the \\
   \> \> \> \> $p-1$ splitters of ${S}$ ; \\
12.\> \> \> {\bf for} all $i \in \{0,\ldots,p-1\}$ \\
13.\> \> \> \> {\bf do} \> {\bf communicate} subsequence $X^{\langle k
              \rangle}_{i}$ to $Z^{\langle i \rangle}_{k}$ ; \\
14.\> \> \> {\bf let} $\bar{X}^{\langle k \rangle} =
              \mbox{{\sc Merge\_Seq}($\bigcup_{i}Z^{\langle k
              \rangle}_{i}$)}$ ; \\
15.\> \> \> {\bf return} $\bar{X}^{\langle k \rangle}$ ; \\
{\bf end} {\sc Sort\_IRan\_BSP}
\end{tabbing}
\end{minipage}
}
}
\caption{{\small Procedure {\sc Sort\_IRan\_BSP}.}}
\label{SortRI}
\end{figure}

Compared to
{\sc Sort\_Ran\_BSP}, in {\sc Sort\_IRan\_BSP} local sorting is performed
first on a set of $n/p$ keys, not $(1+ 1/\omega_n)n/p$.
Binary search of the input keys into the splitters can be simplified
by merging the two in linear time or performing a binary search of the
splitters into the input key set, a less expensive operation; the latter
operation was implemented in the code of our experiments. 
Communication is simpler, as by the initial 
sorting and the binary-search operation of the splitter keys
into the sorted local input keys,
each processor is able to communicate a contiguous block of 
keys to every other processor.  Because each processor receives
sorted sequences a multi-way merge operation
at the conclusion of the algorithm is required for the same asymptotic 
cost in terms of comparisons performed to the binary search operation
of the keys into the splitter that has become unnecessary.

Computation time of {\sc Sort\_IRan\_BSP} is
$n/p \lg{(n/p)}+ (1+1/\omega_n) n \lg{p}/p + 2 \omega_n^2 \lg{n}\lg^2{p}+
O(p \lg{(n/p)}+\omega_n^2 \lg{n} \lg{p})$ and
communication time is 
$(1+1/\omega_n )ng/p +g \omega_n^2 \lg{n} \lg^2{p}+ L\lg^2{p}/2 +
O(L\lg{p}+ g \omega_n^2 \lg{n}\lg{p}+pg)$. 
The first two terms of computation time are due to local sorting and
multi-way merging respectively, and the 
third term is due to parallel sample sorting. 
The first term of communication time is due to key
routing, the second term is due to parallel sample sorting, and the
third term reflects the number of synchronization operations required
for parallel (Batcher-based) sample sorting.
The terms hidden in the $O(\cdot)$ notation describe contributions of the
remaining auxiliary operations. For example, in computation time the first
term in $O(\cdot)$ reflects the cost of binary search of the splitters
into the local keys, and the second term low order contributions of parallel
(Batcher based) sample sorting.
Similarly for communication time the first and seconds terms in $O(\cdot)$ 
reflect lower order contributions in parallel sample sorting and the third
term the cost of splitter broadcasting.
Therefore for $p^2 \leq n /(\omega_n \lg{n}) $  
and $2p \omega_n^2 \lg{p} < n/2$,
and $L \leq 2n /(p \lg^2{p})$, we conclude that
$\pi = 1+ \lg{p}/(\omega_n \lg{n}) + 2p\omega_n^2 \lg^2{p}/n
+O(1/\omega_n \lg{n}+\omega_n^{3/2} \lg{p} / \sqrt{n \lg{n}} )$ 
and
$\mu = (1+1/\omega_n )g/\lg{n}+gp\omega_n^2 \lg^2{p}/n+ 
L p \lg^2{p}/(2n\lg{n})+ O(1/\lg{p}\lg{n}+
g\omega_n^{3/2}\lg{p}/\sqrt{n\lg{n}} +g /(\omega_n \lg^2{n}) )$. 
Proposition \ref{iransorting} is then derived.

\begin{proposition}
\label{iransorting}
For any $\omega_n$ such that $2p \omega_n^2 \lg{n}<n/2$, 
$p^2 \leq n / (\omega_n \lg{n})$ and $L \leq 2n /(p \lg^2{p})$, 
algorithm {\sc Sort\_IRan\_BSP} is  such that 
$\pi = 1+ \lg{p}/(\omega_n \lg{n}) + 2p\omega_n^2 \lg^2{p}/n
+O(1/\omega_n \lg{n}+\omega_n^{3/2} \lg{p}/ \sqrt{n \lg{n}} )$ 
and
$\mu = (1+1/\omega_n )g/\lg{n}+gp\omega_n^2 \lg^2{p}/n+ 
L p \lg^2{p}/(2n\lg{n})+O(1/\lg{p}\lg{n}+
g\omega_n^{3/2}\lg{p}/\sqrt{n\lg{n}} +g /(\omega_n \lg^2{n}) )$. 
\end{proposition}

%As both the deterministic and randomized  implementations require
%Batcher's bitonic-sort algorithm  \cite{Batcher68}, 
%we have implemented this algorithm as well.
%The variant implemented is the one described in 
%\cite{GS96f}, where comparators are replaced by $(n/p)$-way mergers 
%\cite{Knuth73}.  For details we refer to \cite{GS96f}. 

\section{Performance Evaluation}
\label{s6}

As the BSP model is not just an abstract architectural model, it may also
be employed as a programming platform or, indeed, as a kind of a
programming paradigm.  The underlying concept
in the BSP model is the notion of the superstep and the abstraction that
non-local communication associated with a superstep takes place between
supersteps as a global operation.  Thus, a BSP program may be viewed as a
succession of supersteps, with the required non-local communication occurring
at the end of each superstep. Viewed  this way, the BSP model can be
realized as a library of functions with architecture independent semantics
for process creation, remote data access and bulk synchronization.
The Oxford BSP Toolset, \bsplib{}, implements such a paradigm \cite{Hill96}
and provides library support for BSP programming. 
The library functions are callable from standard imperative 
languages such as C and  Fortran. 

It should be borne in mind that 
such support can be made available by other non BSP-specific libraries
(e.g. MPI) that support the simple communication and synchronization 
primitives required for programming under the BSP model.
The effort required to learn the basics of \bsplib{} is minimal.
One needs to understand the semantics of no more than 10-15 functions;
half of these functions are related to process creation, 
destruction and identification. 
The implementations presented in this work were 
programmed in ANSI C. Only recompilation is required to run the same 
code on other platforms such as Silicon Graphics Origin 2000 and Power
Challenge systems, or an IBM SP2 system.
Our implementations are test for scalability 
and portability on a  distributed memory system, a 128-processor 
Cray T3D. 
The manufacturer-supplied C compiler ({\tt cc}) 
is used through the \bsplib{} front-end, and the source-code is 
compiled with the {\tt -O3} compiler option set.
Timing is obtained through the use of the real-time (wall-clock time) 
clock function {\tt bsp\_time} of \bsplib{} \cite{BSP}. 
The timing results obtained are discussed later in this section.
The depicted results in the following tables are in general,
averages over at least four experiments.
In the following discussion we shall require the BSP
parameters of the machine configurations test. 
%In an earlier version of this work as established in \cite{G97} these 
%values were obtained from \cite{Hill97}. Those values, 
%however, were not reliable
%and in some cases significantly underestimated the true performance 
%of the version of \bsplib{} used in the experiments. The 
%underestimation of the BSP
%parameters coupled with the overestimation of communication time
%under the BSP cost model led to a good match between predicted  (under
%the BSP cost model) and observed (by our experiments) running time.
%These problems have been reported in detail in \cite{GP}.
%In this work more representative values for the BSP parameters
%for the test machines when modeled as BSP computers are used \cite{GP}.
The CRAY T3D is thus reported to behave as a BSP machine with
sets of parameters 
$(16, 130\mu sec, 0.21\mu sec/int)$, 
$(32, 175\mu sec, 0.26\mu sec/int)$, 
$(64, 364\mu sec, 0.28\mu sec/int)$, 
$(128,762\mu sec, 0.34\mu sec/int)$,
for the configuration used in the experiments
(data type in communication is a 64-bit integer).
Our implementation of quicksort, sorts $1024\times 1024$ integer keys in 
about 3 seconds. This is equivalent to 7 comparisons per microsecond;
expressing $g$ in terms of basic computational operations (i.e. comparisons)
a value of  $g$ of $0.21 \mu sec/int$ becomes $0.21\times 7\approx 1.47$ 
comparisons/int.
The BSP approach in modeling the performance and running time
characteristics of parallel algorithms and programs is by modeling
some abstract features of the underlying parallel platform as expressed
through the parameters $p$, $L$ and $g$. This parametrization is not
as pure in terms of measuring hardware performance as say that of the LogP
model. In fact it may be considered as more flexible, accommodating
and forgiving. Parallel performance prediction has also attempted to
model high-level parallel and sequential operations  \cite{Dusseau}. 
The problem with such approaches as also reported in \cite{Dusseau} is
that computational operation modeling can be affected by other
platform characteristics (e.g. caching) 
thus making the study  of scalability issues difficult.
The minimalist approach of the BSP model may
lead to fewer such problems, in general.

%In all of the algorithms we employed the BSP model in both the design and
%implementation stages of the algorithm development.  That is, we based
%the design and implementation choices purely on the BSP cost model, i.e.,
%the algorithms take into explicit consideration the BSP parameters $p$,
%$g$ and $L$ of the implementation platform.
As the objectives of algorithm design on the BSP model are 
{\it scalability}, {\it portability} and {\it predictability of 
performance}, the algorithms we suggest may not lead to the best 
possible implementation on a specific platform. 
The sequential methods used may not be the best possible.
One may also improve performance (sequential and parallel) by directly 
using specific features and interfaces of each machine 
(e.g. communication primitives). 
Our aim is to show that all three objectives can be realized without 
significant loss of performance in a general, architecture independent way. 
This is similar to the portability/reusability one obtains by programming 
in a high-level language as opposed to assembly.

%%HERE
\subsection{Algorithm Implementation Features}

In the implementation of the BSP algorithms, the following choices have
been made.

%The reasoning behind the choice of the standard C library function {\tt qsort} 
%is that {\tt qsort} is a widely available generic sorting routine, and thus, 
%facilitates the comparison of parallel sorting algorithms that use it as the 
%underlying sequential sorting algorithm.  The term {\em generic} implies that 
%{\tt qsort} can be used to sort {\em any} data type.  To this end, a 
%comparison function {\tt compare} is used to facilitate the comparison of two 
%keys $e_1$ and $e_2$.  In particular, the {\tt compare} function returns 
%1, 0 or -1 depending on whether $e_1$ is greater, equal or smaller than $e_2$. 

\begin{ilist}
\item The implemented parallel algorithms are {\em comparison-based} sorting
algorithms even though the data type of the input keys 
is ANSI C {\tt int} integers, and one of the two methods
used for sequential sorting is  an integer-specific method, radixsort.
We use integer data keys for input as most other experimental studies 
also sort integer data and we would like to be able to compare our
implementations to other ones. In addition, comparisons performed on integers 
are faster than those performed on strings of arbitrary length or 
non-trivial structures. This way we make sequential operations as fast 
as possible for the purpose of showing the parallel (communication) 
efficiency of our designs and implementations.

\item Our algorithms {\em naturally handle duplicate keys} 
efficiently. The approach to handle duplicate keys is the one
discussed in Section \ref{duplicate}. We used it for both the
deterministic and randomized algorithm implementations.
As a result, the introduced algorithms are independent of input 
distribution.  As it was discussed in Section \ref{duplicate},
the memory overhead of handling duplicate keys is negligible, and
may triple in the worst case the sample size as it attaches
to each sample key an integer processor identifier and an integer
array index; as sample size is $o(1)$ of problem size such overhead
is tolerable.
The overhead of duplicate handling in computation
and communication time is in general asymptotically negligible
thus affecting only lower order terms of the running time
(a 3-6\% deterioration in performance was observed in most of
the experiments compared to test cases where duplicate key
handling was intentionally switched off). 
The effect of duplicate key handling becomes non-negligible when 
sorting 1M integers (relatively few keys) on 128 (many) processors. 
Our approach is  in contrast to the methods of 
\cite{jaja1,jaja2} that require twice as much communication to 
accommodate duplicate keys.

\item 
Given the scarcity of experimental study results that can
be used for comparison purposes and the fact that experimental sorting 
studies (e.g. \cite{jaja1,jaja2}) with integer input data use radixsort
for sequential sorting, our implementations and experimental studies 
use two different algorithms for sequential sorting.
One of the sequential sorting algorithms is a purely comparison-based 
algorithm, an author-written version of {\em quicksort}, and another one 
is an author-written integer specific version of {\em radixsort}. 
Radixsort was implemented purely for the purpose of comparing our 
implementations with other comparison-based parallel algorithms 
(such as the ones in \cite{jaja1,jaja2,jaja3}) that nevertheless 
use radixsort for sequential sorting of integers.

\item  Previous experimental results \cite{GS97a} describe generic 
implementations of our algorithms on the test and other platforms. 
A generic implementation is one that can generically sort any data type. 
This is achieved by including in the parameters of the parallel sorting 
function an additional argument, which is a function called {\tt compare}, 
that compares two instances of the data type that is used as input for 
sorting. Standard C library function {\tt qsort} is such an example of a
generic sorting function.  A function  like {\tt compare} returns -1, 
0, or +1 depending on whether the former of its arguments is less, equal or 
greater than the latter one. The performance of a generic sorting algorithm
that performs $O(n \lg{n})$ comparisons to sort $n$ keys is then bound 
by the $O(n \lg{n})$  number of function calls to {\tt compare}. As a
consequence, a generic sorting function becomes 4-7 times slower than 
the corresponding
non-generic one as it is also evident in the comparison of \cite{jaja1}
(Table X) that compares a non-generic to a generic implementation.
As all other experimental parallel sorting studies use
non-generic implementations, we decided to follow their approach
in this study.  For a non-generic implementation however, one needs to 
create a new set of functions for each data type used. We 
therefore decided to implement non-generic functions of the otherwise 
originally generic code that operate on integers.

\item Total sample size over all the processors for the deterministic
algorithm is $p^2 \lceil \omega_n \rceil$, where $\omega_n = \lgg{n}$.
For the randomized algorithm it is $2p \omega_n^2 \lg{n}$, where
$\omega_n^2 = \lg{n}$.
\end{ilist}

\subsection{Algorithm implementations}

The following BSP sorting algorithms have been implemented on top of 
\bsplib{}. 

\begin{elist}
\item Two variants  of the one-optimal deterministic algorithm 
      {\sc Sort\_Det\_BSP}  that both operate on ANSI C {\tt int} data 
      types, with duplicate key handling performed according to the method
      described in Section \ref{duplicate}. One uses for sequential sorting, 
      quicksort (an author written 
      implementation) and is called {\sf [DSQ]}, and the other uses 
      radixsort and called {\sf [DSR]}. 
\item Two variants of the one-optimal randomized  algorithm 
      {\sc Sort\_IRan\_BSP}, with duplicate key handling performed according
      to the method described in Section \ref{duplicate},
      called, similarly to the previous case, {\sf [RSQ]} and {\sf [RSR]}.
\item A version of Batcher's bitonic sort \cite{Batcher68} has been 
      implemented and is called {\sf [BSI]}. {\sf [BSI]} is used for parallel 
      sample sorting only. As its performance is worse than that of any of 
      the other four implementations in all but very small problem and
      processor sizes 
      (for such cases, Batcher's algorithm is faster because of its low 
       overhead) we do not compare its performance to our other implementations.
\end{elist}

\noindent 
{\bf Remark 1.}  The implementations can handle duplicate
keys as this was previously explained.

\medskip
\noindent
{\bf Remark 2.}  The implementations  are 
{\em portable} and {\em reusable} enough to allow any 
sequential sorting or merging algorithm to be used as the 
underlying sequential method. 

\subsection{Sorting Benchmarks}

We have test our implementations on a variety of input sets. We tried to 
conform to previously published data sets \cite{jaja1,jaja2,jaja3} 
by including them in this experimental study.  
For a discussion of the practicality and suitability of this set of 
sorting benchmarks we  refer to \cite{jaja3,jaja2,jaja1}.
The sorting benchmarks are briefly defined and
{\tt INT\_MAX} is the maximum integer value plus one
accommodated in a 32-bit signed arithmetic data type (e.g., $2^{31}$).  
The definitions below follow those appearing in \cite{jaja3}, except for
the worst regular input set that follows the definition in 
\cite{jaja1,jaja2}. The
format of the tables reporting the results is similar to the ones
of \cite{jaja1,jaja2,jaja3}. In addition,
we decided not to test our implementation on two additional
data sets of \cite{jaja1,jaja2}, i.e.  {\sf [Z]} and {\sf [RD]} for
one good reason. Due to the way our duplicate-key handling method
works, some preliminary results for these two distributions were in no
case  worse than that of {\sf [U]} (in fact, they were better) 
and similar to those of {\sf [DD]} or {\sf [WR]}. These observations
agree with the results of \cite{jaja1} where {\sf [Z]} and {\sf [RD]}
give  results similar to {\sf [DD]} and no worse than those of {\sf [U]},
and of \cite{jaja2} where the results of  {\sf [Z]} and {\sf [RD]}
are similar to {\sf [DD]} and no worse than those of {\sf [U]}.

\begin{elist}
\item Uniform {\sf [U]}, the input is uniformly and at random distributed,
      and is  generated by calling a pseudo random number generator, the C 
      standard library function {\tt random()}, which returns a long 
      (integer) in the range $[0, \ldots, 2^{31}-1]$ and processor's $i$
      seed is $21+1001 \cdot i$.  
\item Gaussian {\sf [G]}, the input follows a Gaussian distribution, 
      and is approximated by adding the results of four calls to 
      {\tt random()} and dividing the sum by four.
\item Bucket Sorted {\sf [B]}, the input (per processor) is split into $p$ 
      buckets, each of size $n/p^2$, so that the $i$-th such bucket, 
      $0 \leq i \leq p-1$, consists of numbers uniformly and at random 
      distributed in the range $[i\mbox{\tt INT\_MAX}/p, \ldots, 
      (i+1)\mbox{\tt INT\_MAX}/p-1]$.
\item $g$-Group {\sf [g-G]}, the processors are first divided into $p/g$ 
      groups each consisting of $g$ processors, and within group $j$ each
      processors splits its input into $g$ buckets so that the $i$-th such
      bucket consists of numbers uniformly and at random distributed in the 
      range  $[((jg+p/2+i) \bmod p) \mbox{\tt INT\_MAX}/p, \ldots, 
      ((jg+p/2+i+1) \bmod p) \mbox{\tt INT\_MAX}/p-1]$.
\item Staggered {\sf [S]}, the input of processor $i < p/2$ is uniformly 
      and at random distributed in the range $[(2i+1)\mbox{\tt INT\_MAX}/p, 
      \ldots, (2i+2)\mbox{\tt INT\_MAX}/p-1]$, and of processor $i \geq p/2$ 
      in the range $[(i-p/2)\mbox{\tt INT\_MAX}/p, \ldots, (i-p/2+1)
      \mbox{\tt INT\_MAX}/p-1]$.
\item Deterministic Duplicates {\sf [DD]}, the $n/2^{i}$ input keys
      of the first $p/2^i$ consecutive processors are all set to 
      $\lg{(n/p^{i-1})}$, and so forth. In processor $p-1$, $n/(p 2^i)$ keys, 
      starting from the lower indexed and proceeding to higher indexed, are 
      set to $\lg{n/(p 2^{i-1})}$ and so forth.
\item Worst Regular {\sf [WR]}, as described in \cite{jaja1}.
\end{elist}

\subsection{Experimental Results}

In this section, we report on the performance and relative merits of the 
implementations of the algorithms of Section~\ref{UAlg}.
%and the experience gained from studying the 
%predicted theoretical performance of these algorithms as opposed to their 
%actual running time.  
The tables below summarize some of the experimental results we had obtained. 
%It is noted that no particular effort was taken to optimize the 
%implementations (other than using the optimization switches of the compiler) 
%as our purpose has been to write comparatively readable and reusable code.
Timing figures are in general truncated to three or fewer decimal digits.

Table \ref{tbl1} shows timing results for algorithm
{\sc Sort\_IRan\_BSP} on 64 processors of a Cray T3D. 
Results for both {\sf [RSR]} and {\sf [RSQ]} are presented
for all seven benchmark input sets. Size refers to the total size 
$n$ of the input over all 64 processors.

It is possible to compare our implementation of {\sc Sort\_IRan\_BSP}
to the one described and implemented in \cite{jaja3} (Table I, page 215). 
For problem size 1M, ours is slower by about 3-10\%. For the other common
problem sizes (4M, 16M, 64M) our implementation is faster by about 3-8\%.
It is worth noting that for small problem sizes (e.g. 1M)
the use of quicksort for sequential sorting improves the parallel
performance of our implementation by 1-2\%.

Compared to the implementation of the randomized sorting algorithm introduced
in \cite{jaja2}, our implementation is slower by  3-15\% for problem
size 1M, 3-5\% for problem size 4M, and about 1-2\% for the remaining problem
sizes. The slowness of our implementation is mainly attributable to
sequential merging which takes 33-39\% of the total execution time 
of any one experiment
as opposed to only 25\% for \cite{jaja2}. Sequential operations (sorting 
and merging related) take together 90\% of execution time in our 
implementation and 80\% in the implementation of \cite{jaja2}. Although our 
implementation is more communication efficient, mainly because it uses only 
one round of data communication, this advantage is wasted by the apparent 
inefficiency of a sequential operation (merging).

\begin{table}[htb]\scriptsize \centering
\begin{tabular}{|r|r|r|r|r|r|r|r||r|r|r|r|r|r|r|} \hline 
  \multicolumn{15}{|c|}{Running Time of {\sc Sort\_IRan\_BSP} on 64 procs} \\ \hline
  \multicolumn{1}{|c|}{}      &
  \multicolumn{7}{|c||}{[RSR]} &
  \multicolumn{7}{|c|}{[RSQ]} \\ \cline{2-15}
Size& {\sf [U]}&{\sf [G]}&{\sf [2-G]}&{\sf [B]}&{\sf [S]}&{\sf [DD]}&{\sf [WR]}& 
      {\sf [U]}&{\sf [G]}&{\sf [B]}&{\sf [2-G]}&{\sf [S]}&{\sf [DD]}&{\sf [WR]} \\ \hline
1M       &0.079 &0.077 &0.073 &0.079 &0.062 &0.053 &0.079 &
	  0.076 &0.078 &0.072 &0.068 &0.062 &0.050 &0.067 \\ \hline
4M       &0.269 &0.269 &0.254 &0.270 &0.211 &0.171 &0.271 &
	  0.283 &0.280 &0.258 &0.239 &0.234 &0.176 &0.241 \\ \hline
8M       &0.526 &0.526 &0.494 &0.527 &0.408 &0.329 &0.528 &
	  0.559 &0.563 &0.522 &0.484 &0.462 &0.356 &0.486 \\ \hline
16M	 &1.03  &1.03  &0.987 &1.04  &0.802 &0.643 &1.04  &
          1.15 	&1.16  &1.10  &1.00  &0.960 &0.735 &1.00 \\ \hline
32M 	 &2.06	&2.07  &1.92  &2.06  &1.60  &1.27  &2.06  &
          2.39	&2.37  &2.24  &2.05  &2.00  &1.54  &2.06  \\ \hline
64M 	 &4.09	&4.09  &3.90  &4.09  &3.16  &2.50  &4.11  &
	  4.88	&4.88  &4.65  &4.29  &4.16  &3.17  &4.31 \\ \hline
%1M       &0.079 &0.077 &0.073 &0.079 &0.062 &0.053 &0.079 &
%	  0.076 &0.078 &0.072 &0.068 &0.062 &0.050 &0.067 \\ \hline
%4M       &0.269 &0.269 &0.254 &0.270 &0.211 &0.171 &0.271 &
%	  0.283 &0.280 &0.258 &0.239 &0.234 &0.176 &0.241 \\ \hline
%8M       &0.526 &0.526 &0.494 &0.527 &0.408 &0.329 &0.528 &
%	  0.559 &0.563 &0.522 &0.484 &0.462 &0.356 &0.486 \\ \hline
%16M	 &1.03  &1.034 &0.987 &1.041 &0.802 &0.643 &1.042 &
%          1.15 	&1.16  &1.10  &1.00  &0.960 &0.735 &1.00 \\ \hline
%32M 	 &2.06	&2.07  &1.92  &2.06  &1.60  &1.27  &2.06  &
%          2.39	&2.37  &2.24  &2.05  &2.00  &1.54  &2.06  \\ \hline
%64M 	 &4.09	&4.09  &3.90  &4.09  &3.16  &2.50  &4.11  &
%	  4.88	&4.88  &4.65  &4.29  &4.16  &3.17  &4.31 \\ \hline
\end{tabular}
\caption{{\small Execution time of {\sc Sort\_IRan\_BSP} with $p=64$ 
($1M=1024\times 1024$)}.}
\label{tbl1}
\end{table}

\begin{table}[htb] \scriptsize \centering
\begin{tabular}{|r|r|r|r|r|r|r|r||r|r|r|r|r|r|r|} \hline 
  \multicolumn{15}{|c|}{Running Time of {\sc Sort\_Det\_BSP} on 64 procs} \\ \hline
  \multicolumn{1}{|c|}{}      &
  \multicolumn{7}{|c||}{[DSR]} &
  \multicolumn{7}{|c|}{[DSQ]} \\ \cline{2-15}
Size& {\sf [U]}&{\sf [G]}&{\sf [2-G]}&{\sf [B]}&{\sf [S]}&{\sf [DD]}&{\sf [WR]}& 
      {\sf [U]}&{\sf [G]}&{\sf [B]}&{\sf [2-G]}&{\sf [S]}&{\sf [DD]}&{\sf [WR]} \\ \hline
1M        &0.091	&0.091	&0.084	&0.089	&0.073	&0.063	&0.089 &
	   0.088	&0.088	&0.083	&0.079	&0.072	&0.060	&0.079 \\ \hline
4M        &0.281	&0.280	&0.265	&0.279	&0.218	&0.177	&0.280 &
	   0.294	&0.291	&0.269	&0.249	&0.242	&0.183	&0.250 \\ \hline
8M        &0.532	&0.530	&0.507	&0.529	&0.415	&0.330	&0.529 &
	   0.566	&0.566	&0.535	&0.487	&0.470	&0.360	&0.489 \\ \hline
16M       &1.03		&1.03	&0.985	&1.03	&0.803	&0.633	&1.03 &
	   1.15		&1.16	&1.09	&0.985	&0.963	&0.727	&0.990\\ \hline
32M       &2.06		&2.06	&1.92	&2.07	&1.58	&1.24	&2.07 &
	   2.38		&2.37	&2.26	&2.06	&1.98	&1.50	&2.07 \\ \hline
64M       &4.09		&4.09	&3.88	&4.08	&3.12	&2.44	&4.08 &
	   4.94		&4.86	&4.63	&4.27	&4.18	&3.14	&4.27 \\ \hline
\end{tabular}
\caption{{\small Execution time of {\sc Sort\_Det\_BSP} with $p=64$ 
($1M=1024\times 1024$)}.}
\label{tbl2}
\end{table}

The experimental results obtained for all seven benchmark sets
for algorithm {\sc Sort\_Det\_BSP} are depicted in
Table \ref{tbl2}. A comparison of our implementation to the one
in \cite{jaja3} (Table IV, page 218) indicates that ours 
is faster by as much as 13-35\% 
for problem size 1M, 15-25\% for problem size 4M, 12-15\% for 16M, and 
7-15\% for 64M.
This is possibly attributable to the fact that the algorithm in \cite{jaja3}
is a direct variant of the one in \cite{SH}. Our implementation on the
other hand, although it is also based on ideas of \cite{SH}, 
it naturally handles duplicate keys (without any significant 
increase of communication time), implements deterministic oversampling and 
determines sample 
size on the basis of maximum key imbalance at the completion of sorting.
Its design and analysis is solely based on the BSP model
and seems to perform better in practice exhibiting performance that is
within the predictions of the theoretical analysis.

The algorithm of \cite{jaja3} was further refined into an oversampling-based 
algorithm  that takes into account duplicate keys  and
is described in \cite{jaja1}. This latter algorithm also handles duplicate keys
by performing twice as much communication.
Comparing the performance of that latter algorithm (Table III of \cite{jaja1})
to our results as obtained in Table \ref{tbl2} we observe that
for problem size 1M our implementation is faster
by as much as 10-20\%, and for problem sizes, 4M, 16M and 64M our algorithm
is faster by 8-26\%, 3-10\% and 2-3\%. In our implementation sorting
requires one communication round instead of two in \cite{jaja1}. 
Even though our sequential methods are slower (as exhibited in
Table \ref{tbl6r}) than the ones implemented in \cite{jaja1} 
(Table IX), the extra communication round
of the algorithm in \cite{jaja1} neutralizes   such an advantage 
even though the cost of a communication round is no more 
than 5-6\% of the overall time of the algorithm as opposed to sequential 
sorting and merging that 
account for 45-60\% and 30-40\% of the running time respectively.

\begin{table}[htb] \scriptsize \centering
\begin{tabular}{|r|r|r|r|r|r||r|r|r|r|r|} \hline 
  \multicolumn{1}{|c|}{}      &
  \multicolumn{5}{|c||}{[RSR]} &
  \multicolumn{5}{|c|}{[RSQ]} \\ \cline{2-11}
Input Set&{$p=8$  }&{$p=16$ }&{$p=32$ }&{$p=64$ }&{$p=128$}& 
          {$p=8$  }&{$p=16$ }&{$p=32$ }&{$p=64$ }&{$p=128$} \\ \hline
{\sf [U]} &3.16	&1.74	&0.956	&0.526	&0.300 (65\% )&
           3.90	&2.00	&1.07	&0.559	&0.310 (78\% )\\ \hline
{\sf [WR]}&3.16	&1.74	&0.956	&0.527	&0.306 (64\% )
          &3.64 &1.82	&0.938	&0.486	&0.272 (83\% )\\ \hline
  \multicolumn{1}{|c|}{}      &
  \multicolumn{5}{|c||}{[DSR]} &
  \multicolumn{5}{|c|}{[DSQ]} \\ \cline{2-11}
{\sf [U]} & 3.18 &1.72	&0.947	&0.532	&0.374 (53\% )
          & 3.92 &1.98	&1.07 	&0.566	&0.386 (63\% )\\ \hline
{\sf [WR]}& 3.18 &1.73	&0.945	&0.530	&0.372 (53\% )
          & 3.65 &1.82	&0.930	&0.489	&0.337 (67\% )\\ \hline
\end{tabular}
\caption{{\small Execution time of {\sc Sort\_IRan\_BSP} and 
{\sc Sort\_Det\_BSP} on an input of size 
$8\times 1024\times 1024$}.For $p=128$ efficiencies are also shown.}
\label{tbl3}
\end{table}

Table \ref{tbl3} shows the scalability of our algorithm implementations
on inputs {\sf [U]} and {\sf [WR]}. We used these two input distributions
as they have also been used in the experimental work of 
\cite{jaja1,jaja2,jaja3}; in those studies, for a deterministic sorting 
algorithm, {\sf [WR]} is used exclusively \cite{jaja1,jaja3}, 
whereas for a randomized 
sorting algorithm both {\sf [WR]} \cite{jaja3} and {\sf [U]} \cite{jaja2} 
are used.
A scalability comparison of all these algorithms
is depicted in Table \ref{tbl9} later in this work.
With respect to Table \ref{tbl3} it is worth noting that for relatively
small ratios $n/p$ (e.g. when $p=128$ and for
moderate problem sizes), sequential sorting of integers using quicksort
may yield better performance than using radixsort, as expected.
As the variants of our algorithms that use quicksort for sequential sorting 
are more CPU intensive, efficiencies are higher for such cases than those 
that use radixsort. For $p=128$, {\sf [RSR]} has the same efficiency 64-65\%
independent of the input data distribution; for {\sf [RSQ]} this varies 
between 78\% to 83\%. The efficiencies of the deterministic algorithm are 
lower however; they are 53\% and vary between 63\% - 67\%, i.e. they are
between 12-15\% lower than those of the randomized  algorithm. 

This is due to the fact that for a fixed size sample, random sampling 
yields more balanced set of keys in the routing and multi-way merging phases 
of the sorting algorithm. 
We note that for the deterministic algorithm we chose $\omega_n = \lgg{n}$
and for the randomized algorithm $\omega_n^2 = \lg{n}$. 
In all runs of our experiments (of either
the deterministic or the randomized algorithm) maximum set
imbalance was kept below 15\% well within the approximately 20\% obtained
from the theoretical analysis for the deterministic 
($1/\ceil{\lgg{n}} \times 100\%$) and randomized 
algorithms ($1/\sqrt{\lg{n}} \times 100 \%$). 

For problem size $n=2^{23}=8M$ as used in Table \ref{tbl3} and for $p=128$, 
the conditions on $n$, $p$ and $L$ in Proposition \ref{detsorting} and 
\ref{iransorting} are hardly satisfied.
Substituting for $n$, $p$, $L$ and $g$ in the expressions for
$\pi$ and $\mu$ in Proposition \ref{detsorting}  and ignoring the low
order terms (enclosed in $O(.)$) 
we derive a theoretical bound on efficiency of at least 66\% for {\sf [DSQ]}.
The observed efficiency was 63\% and 67\% for {\sf [U]} and {\sf [WR]} 
respectively. We note that the contribution to the theoretical estimation of
running time of handling duplicate keys was ignored (such contributions
would have been absorbed in the $O(.)$ term which is anyway ignored in
the calculation of the theoretical efficiency). Had we disabled the
code for handling duplicate keys, for $p=128$ running time for
{\sf [DSQ]} on input {\sf [U]} would be 0.348 (i.e. efficiency 70\% instead
of 63\%) and that for {\sf [DSR]} on input {\sf [U]} would be 0.336 
(i.e. efficiency  58\% instead of 53\%).
For the randomized algorithm the theoretical prediction of
at least 66\% was also satisfied in practice 
(observed efficiency was 78-82\% ).
We note that the theoretical analysis applies only to the case that sequential 
sorting is performed by a comparison and exchange algorithm only. 
Although the execution time of radixsort is independent of the
input distribution (if we ignore caching effects), that of quicksort
is not. 

\begin{table}[htb] \scriptsize \centering
\begin{tabular}{|r|r|r|r||r|r|r||r|r|r||r|r|r|} \hline 
  \multicolumn{1}{|c|}{}      &
  \multicolumn{6}{|c||}{Time per phase} &
  \multicolumn{6}{|c|}{Percentage (\%) of total time for each phase} \\ \cline{2-13}
  \multicolumn{1}{|c|}{}      &
  \multicolumn{3}{|c||}{8M} &
  \multicolumn{3}{|c||}{32M} &
  \multicolumn{3}{|c||}{8M} &
  \multicolumn{3}{|c|}{32M} \\ \cline{2-13}
Procs&32       &64      &128     &32      &       64&   128 
     &32       &64      &128     &32      &       64&   128  \\ \hline
Ph 1 &  0.000  & 0.001  & 0.002  & 0.000  & 0.001  & 0.002  & 0.07  & 0.23  & 0.73  & 0.02  & 0.06  & 0.20  \\ \hline
Ph 2 &  0.560  & 0.277  & 0.137  & 2.220 & 1.115   & 0.557  & 57.76 & 52.64 & 45.48 & 58.62 & 54.09 & 49.49 \\ \hline
Ph 3 &  0.011  & 0.011  & 0.017  & 0.011  & 0.013   & 0.020  & 1.13  & 2.11  & 5.61  & 0.29  & 0.63  & 1.77  \\ \hline
Ph 4 &  0.004  & 0.003  & 0.006  & 0.004  & 0.003   & 0.006  & 0.41  & 0.57  & 1.98  & 0.11  & 0.15  & 0.53  \\ \hline
Ph 5 &  0.066  & 0.036  & 0.021  & 0.259  & 0.144   & 0.080  & 6.80  & 6.82  & 6.93  & 6.83  & 6.98  & 7.10  \\ \hline
Ph 6 &  0.324  & 0.198  & 0.118  & 1.288  & 0.786   & 0.461  & 33.40 & 37.57 & 38.94 & 33.98 & 38.10 & 40.91 \\ \hline
Ph 7 &  0.004  & 0.000  & 0.001  & 0.006  & 0.000   & 0.000  & 0.41  & 0.06  & 0.33  & 0.16  & 0.00  & 0.00  \\ \hline
Total&  0.970  & 0.527  & 0.303  & 3.791  & 2.063   & 1.127  & 100   & 100   & 100   & 100   & 100   & 100   \\ \hline
\end{tabular}
\caption{{\small Scalability of various phases of {\sf [RSR]} on input
{\sf [U]}. Ph1=Init, 
Ph2=SeqSort, Ph3=Sampling, Ph4=Prefix, Ph5=Routing, Ph6=Merging, 
Ph7=Termination.}}
\label{tbl6r}
%Table 4
\end{table}

\begin{table}[htb]\scriptsize \centering
\begin{tabular}{|r|r|r|r||r|r|r||r|r|r||r|r|r|} \hline 
  \multicolumn{1}{|c|}{}      &
  \multicolumn{6}{|c||}{Time per phase} &
  \multicolumn{6}{|c|}{Percentage (\%) of total time for each phase} \\ \cline{2-13}
  \multicolumn{1}{|c|}{}      &
  \multicolumn{3}{|c||}{8M} &
  \multicolumn{3}{|c||}{32M} &
  \multicolumn{3}{|c||}{8M} &
  \multicolumn{3}{|c|}{32M} \\ \cline{2-13}
Procs&32       &64      &128     &32      &       64&   128 
     &32       &64      &128     &32      &       64&   128  \\ \hline
Ph 1 &  0.000  & 0.001  & 0.002  & 0.000  & 0.001   & 0.002  & 0.06  & 0.25  & 0.76  & 0.02  & 0.05  & 0.18  \\ \hline
Ph 2 &  0.675  & 0.311  & 0.150  & 3.007  & 1.432   & 0.675  & 61.56 & 55.48 & 47.81 & 65.70 & 60.00 & 54.37 \\ \hline
Ph 3 &  0.010  & 0.010  & 0.017  & 0.011  & 0.013   & 0.019  & 0.91  & 1.78  & 5.40  & 0.24  & 0.54  & 1.53  \\ \hline
Ph 4 &  0.014  & 0.003  & 0.005  & 0.005  & 0.003   & 0.006  & 1.28  & 0.53  & 1.59  & 0.11  & 0.13  & 0.48  \\ \hline
Ph 5 &  0.071  & 0.038  & 0.021  & 0.262  & 0.142   & 0.078  & 6.47  & 6.77  & 6.67  & 5.72  & 5.95  & 6.28  \\ \hline
Ph 6 &  0.323  & 0.197  & 0.119  & 1.288  & 0.796   & 0.462  & 29.44 & 35.13 & 37.78 & 28.14 & 33.33 & 37.17 \\ \hline
Ph 7 &  0.003  & 0.000  & 0.000  & 0.003  & 0.000   & 0.000  & 0.27  & 0.05  & 0.00  & 0.07  & 0.00  & 0.00  \\ \hline
Total&  1.097  & 0.561  & 0.315  & 4.577  & 2.388   & 1.243  & 100   & 100   & 100   & 100   & 100   & 100   \\ \hline
\end{tabular}
\caption{{\small Scalability of various phases of {\sf [RSQ]} on input
{\sf [U]}. Ph1=Init, 
Ph2=SeqSort, Ph3=Sampling, Ph4=Prefix, Ph5=Routing, Ph6=Merging, 
Ph7=Termination.}}
\label{tbl6q}
%Table 5
\end{table}

%NOW
Tables \ref{tbl6r} and \ref{tbl6q} show the distribution of execution time for
sorting 8M and 32M integers by the randomized algorithm when input is 
{\sf [U]}. Percentages of total running time for the execution time of each
phase are also shown. With reference to Figure \ref{SortRI} and Tables
\ref{tbl6r}/\ref{tbl6q}, Ph1 corresponds
to the part of the code before line 2, Ph2 (sequential sorting) to the 
execution of lines 2-3,
Ph3 to lines 4-9, Ph4 to lines 9-12, Ph5 (key routing) to lines 12-13, 
and Ph6 (sequential merging) to
lines 14-15 of the code for {\sc Sort\_IRan\_BSP}.
From both tables, the sequential portion of the code contributes
at least 85-90\% of the running time. 
Efficient sequential implementations can thus significantly affect the 
execution time of our implementations,
as it was remarked when a comparison of our implementation to other ones
was made.

%Ph 6 for both variants of the randomized algorithm seems to take the
%same amount of execution time an indication that how input is split does
%not seriously affect the execution time of this Phase. 

It is worth observing that routing time of Ph5 scales satisfactorily
with the number of processors, an indication that the randomized algorithm 
is quite scalable. In the implementation, 
$\omega_n = \sqrt{\lg{n}}$, and therefore in Ph6 each processor is expected 
to send or receive no more than $(1+ 1/ \omega_n ) n /p$ keys. 
The observed imbalance on any of the experiments reported in Tables 
\ref{tbl6r} and \ref{tbl6q} was kept below 15\% which is within the 20\% 
imbalance implied by the theoretical calculations 
(e.g. $1/\sqrt{\lg{2^{23}}} \approx  0.208$).
The communication time of Ph6 gives an observed value for $g$ of
$0.23-0.25$, $0.24-0.28$ and $0.28-0.32$ for $p=32$, $p=64$ and $p=128$
respectively (assuming that any processor sent or received
at most $1.15n/p$ keys) and these values are consistent with the
measured ones (in other experiments \cite{GP98}) of $0.26$, $0.28$ and $0.34$
$\mu sec/int$ respectively reported in the beginning of this section.

\begin{table}[htb]\scriptsize \centering
\begin{tabular}{|r|r|r|r||r|r|r||r|r|r||r|r|r|} \hline 
  \multicolumn{1}{|c|}{}      &
  \multicolumn{6}{|c||}{Time per phase} &
  \multicolumn{6}{|c|}{Percent(\%) of total time per phase} \\ \cline{2-13}
  \multicolumn{1}{|c|}{}      &
  \multicolumn{3}{|c||}{8M} &
  \multicolumn{3}{|c||}{32M} &
  \multicolumn{3}{|c||}{8M} &
  \multicolumn{3}{|c|}{32M} \\ \cline{2-13}
Procs&32       &64      &128     &32      &       64&   128 
     &32       &64      &128     &32      &       64&   128  \\ \hline
Ph 1    &0.000  &0.001   &0.002  &0.000  &0.001   &0.002 &0.07   &0.23  &0.59   &0.02  &0.06  &0.18\\ \hline
Ph 2    &0.560  &0.277   &0.139  &2.222  &1.115   &0.558 &57.94  &53.42 &41.12  &58.64 &54.30 &49.73\\ \hline
Ph 3    &0.005  &0.008   &0.018  &0.006  &0.008   &0.018 &0.52   &1.54  &5.33   &0.16  &0.39  &1.60\\ \hline
Ph 4    &0.004  &0.003   &0.005  &0.005  &0.003   &0.005 &0.41   &0.58  &1.48   &0.13  &0.15  &0.45\\ \hline
Ph 5    &0.073  &0.038   &0.022  &0.260  &0.145   &0.081 &7.55   &7.31  &6.51   &6.86  &7.06  &7.22\\ \hline
Ph 6    &0.315  &0.192   &0.152  &1.294  &0.781   &0.458 &32.57  &36.92 &44.97  &34.14 &38.00 &40.82\\ \hline
Ph 7    &0.009  &0.000   &0.000  &0.002  &0.001   &0.000 &0.93   &0.00  &0.00   &0.05  &0.05  &0.00\\ \hline
Total   &0.967  &0.520   &0.338  &3.79   &2.055   &1.122 &100    &100   &100    &100   &100   &100\\ \hline
\end{tabular}
\caption{{\small Scalability of various phases of {\sf [DSR]} on input
{\sf [U]}. Ph1=Init, 
Ph2=SeqSort, Ph3=Sampling, Ph4=Prefix, Ph5=Routing, Ph6=Merging, 
Ph7=Termination}.}
\label{tbl7r}
% Table 6
\end{table}

\begin{table}[htb]\scriptsize \centering
\begin{tabular}{|r|r|r|r||r|r|r||r|r|r||r|r|r|} \hline 
  \multicolumn{1}{|c|}{}      &
  \multicolumn{6}{|c||}{Time per phase} &
  \multicolumn{6}{|c|}{Percent(\%) of total time per phase} \\ \cline{2-13}
  \multicolumn{1}{|c|}{}      &
  \multicolumn{3}{|c||}{8M} &
  \multicolumn{3}{|c||}{32M} &
  \multicolumn{3}{|c||}{8M} &
  \multicolumn{3}{|c|}{32M} \\ \cline{2-13}
Procs&32       &64      &128     &32      &       64&   128 
     &32       &64      &128     &32      &       64&   128  \\ \hline
Ph 1    &0.000  &0.001   &0.002  &0.000  &0.001   &0.002 & 0.06   &0.22  &0.57   &0.02  &0.05  &0.16\\ \hline
Ph 2    &0.675  &0.310   &0.151  &3.005  &1.433   &0.676 &62.76   &56.2  &43.14  &65.46 &60.32 &54.6\\ \hline
Ph 3    &0.006  &0.008   &0.018  &0.006  &0.008   &0.017 & 0.56   &1.45  &5.14   &0.13  &0.34  &1.37\\ \hline
Ph 4    &0.004  &0.003   &0.005  &0.005  &0.003   &0.005 & 0.37   &0.54  &1.43   &0.11  &0.13  &0.4\\ \hline
Ph 5    &0.071  &0.037   &0.022  &0.268  &0.150   &0.077 & 6.60   &6.69  &6.29   &5.84  &6.31  &6.22\\ \hline
Ph 6    &0.316  &0.192   &0.151  &1.303  &0.780   &0.460 &29.37   &34.72 &43.14  &28.38 &32.81 &37.16\\ \hline
Ph 7    &0.003  &0.001   &0.001  &0.003  &0.001   &0.001 & 0.28   &0.18  &0.29   &0.07  &0.04  &0.08\\ \hline
Total   &1.076  &0.553   &0.350  &4.591  &2.377   &1.238 &100     &100   &100    &100   &100   &100\\ \hline
\end{tabular}
\caption{{\small Scalability of various phases of {\sf [DSQ]} on input
{\sf [U]}. Ph1=Init, 
Ph2=SeqSort, Ph3=Sampling, Ph4=Prefix, Ph5=Routing, Ph6=Merging, 
Ph7=Termination}.}
\label{tbl7q}
% Table 7
\end{table}

\begin{table}[htb]\scriptsize \centering
\begin{tabular}{|l|r|r|r||r|r|r||r|r|r||r|r|r|} \hline 
  \multicolumn{1}{|c|}{}      &
  \multicolumn{6}{|c||}{Time per common phase} &
  \multicolumn{6}{|c|}{Percent(\%) of total time per phase} \\ \cline{2-13}
%  \multicolumn{1}{|c|}{}      &
%  \multicolumn{12}{|c|}{8M} \\ \cline{2-13}
  \multicolumn{1}{|c|}{}      &
  \multicolumn{3}{|c||}{{\sf [DSR]}on {\sf [U]}} &
  \multicolumn{3}{|c||}{\cite{jaja1} on {\sf [WR]}} &
  \multicolumn{3}{|c||}{{\sf [DSR]}on {\sf [U]}} &
  \multicolumn{3}{|c|}{\cite{jaja1} on {\sf [WR]}} \\ \cline{2-13}
     &32       &64      &128     &32      &       64&   128 
     &32       &64      &128     &32      &       64&   128  \\ \hline
Ph 2    &0.560  &0.277   &0.139  &0.591  &0.299   &0.151  &57.94   &53.42 &41.12  &60.60 &50.41 &30.52\\ \hline
Ph R    &  -    &   -    &  -    &0.045  & 0.029  &0.019  & -      & -    &  -    & 4.62 & 4.88 & 3.94 \\ \hline
Ph 5    &0.073  &0.038   &0.022  &0.050  &0.039   &0.076  & 7.55   &7.31  &6.51   & 5.13 & 6.65 &15.35\\ \hline
Ph 6    &0.315  &0.192   &0.152  &0.215  &0.152   &0.107  &32.57   &36.92 &43.97  &22.06 &25.65 &21.64\\ \hline
Total   &0.967  &0.52    &0.338  &0.976  &0.593   &0.496  &100     &100   &100    &100   &100   &100\\ \hline
\end{tabular}
\caption{{\small Scalability comparison of {\sf [DSR]} and \cite{jaja1}. 
Ph2=SeqSort, PhR,Ph5=Routing, Ph6=Merging}.}
\label{tbl8}
% Table 8 comp
\end{table}

Tables \ref{tbl7r} and \ref{tbl7q} show the distribution of execution time for
sorting 8M and 32M integers by the deterministic algorithm when input is 
{\sf [U]}. Percentages of total running time for the execution time of each
phase are also shown.
With reference to Figure \ref{SortD}, Ph1 corresponds
to the part of the code before line 2, Ph2 (sequential sorting) to the 
execution of lines 2-3,
Ph3 to lines 4-7, Ph4 to lines 8-9, Ph5 (key routing) to lines 10-11, 
and Ph6 (sequential merging) to
lines 12-13 of the code for {\sc Sort\_Det\_BSP}.
Communication performance is scalable 
though slightly worse than those of the  randomized algorithm as 
random oversampling causes more balanced communication.

From Tables \ref{tbl6r} and \ref{tbl7r} and Tables \ref{tbl6q} and
\ref{tbl7q}, the sequential portion of the 
code contributes at least 86-93\% of the running time. As the two 
algorithms share most operations except sampling, the performance of
the two algorithms in the first few phases is similar. The quality of 
sampling affects Phases 5 and 6.  This is evident for example from
the figures in the four tables for $n=8M$ and $p=128$. The randomized
algorithm is about $0.033$ seconds faster than the deterministic algorithm
and this is solely due to key-imbalance and the better balancing
properties of the randomized algorithm whose oversampling parameter
can vary more widely than that of the deterministic algorithm.  

Table \ref{tbl8} gives comparative timing results for each one of the
major phases of {\sf [DSR]} and the algorithm in \cite{jaja1} on inputs 
{\sf [U]} and {\sf [WR]} respectively. As for our implementation, 
the running time of {\sf [DSR]} on either input {\sf [U]} or {\sf [WR]} are 
almost identical, the phase by phase timing of our implementation does not 
change when input is {\sf [WR]} rather than  {\sf [U]}. From Table 
\ref{tbl8} we conclude
that even if our sequential methods are slower for multi-way
merging than those of \cite{jaja1}
the balanced single round of communication in {\sf [DSR]}
is more scalable  than the corresponding round
of \cite{jaja1} (listed under Ph 5 in Table \ref{tbl8}),
even though the portable communication library 
of our experiments may be less efficient than the one of 
\cite{jaja1,jaja2,jaja3}.
This may 
be attributable to sampling choices or the way duplicate keys are 
handled in \cite{jaja1}.

\begin{table}[htb] \centering
{\small
\begin{tabular}{|l|r|r|r|r|r|r|} \hline 
Algorithm    & Input     &{$p=8$}  &{$p=16$ }&{$p=32$ }&{$p=64$ }&{$p=128$} \\ \hline
{\sf [RSR]}  &{\sf [U]}  &3.16	   &1.74     &0.956    &0.526    &0.300 \\ \hline
\cite{jaja2} &{\sf [U]}  &3.32	   &1.77     &0.952    &0.513    &0.284 \\ \hline \hline

{\sf [RSR]}  &{\sf [WR]}&3.16	   &1.74     &0.956    &0.527    &0.306 \\ \hline
\cite{jaja3} &{\sf [WR]}&3.21      &1.74     &0.966    &0.540    &0.294\\ \hline \hline

{\sf [DSR]}  &{\sf [WR]}&3.18      &1.73     &0.945    &0.530    &0.372 \\ \hline
\cite{jaja3} &{\sf [WR]}&4.07      &2.11     &1.150    &0.711    &-    \\ \hline
\cite{jaja1} &{\sf [WR]}&3.23      &1.73     &0.976    &0.594    &0.496  \\ \hline \hline

{\sf [DSQ]}  &{\sf [WR]}&3.65      &1.82     &0.930    &0.489    &0.337 \\ \hline
{\sf [RSQ]}  &{\sf [WR]}&3.64      &1.82     &0.938    &0.486    &0.272 \\ \hline \hline
{\sf [DSQ]}  &{\sf [U]} &3.92      &1.98     &1.066    &0.566    &0.386 \\ \hline
{\sf [RSQ]}  &{\sf [U]} &3.90	   &2.00     &1.070    &0.559    &0.310 \\ \hline \hline
{\sf [DSR]}  &{\sf [U]} &3.18      &1.72     &0.947    &0.532    &0.374 \\ \hline
\end{tabular}
\caption{{\small Comparison of our results with other 
algorithm implementations.}}
\label{tbl9}
%Table 9
}
\end{table}

In Table \ref{tbl9} the first four rows compare the performance
of {\sf [RSR]} to the ones in \cite{jaja2,jaja3} for problem size
$n=8M$. Performance 
of the compared
algorithms is similar for up to $p=64$ processors except for $p=128$ where
our implementation is slightly worse by about 5\% attributable mainly to
some inefficiency in sequential merging. The following three rows compare
the performance of {\sf [DSR]} to the deterministic regular sampling 
algorithms in \cite{jaja1,jaja3}.
{\sf [DSR]} is more efficient than the algorithm in \cite{jaja3} through
all processor ranges by about 15-30\%. Compared to the implementation in
\cite{jaja1} the two exhibit about the same performance for up to $p=32$ (the
disadvantage of our sequential merging is counterbalanced by the two
rounds of communication of \cite{jaja2}), and for $p=64$ and $p=128$
{\sf [DSR]} is better by 10\% and 30\% respectively due to more
scalable communication and the single communication round of {\sf [DSR]}.

\begin{table}[htb] \centering
{\small
\begin{tabular}{|r|c|c|c|c|c||c|c|c|c|c|} \hline 
  \multicolumn{1}{|c|}{}      &
  \multicolumn{5}{|c||}{[DSR]} &
  \multicolumn{5}{|c|}{[DSQ]} \\ \cline{2-11}
  \multicolumn{1}{|c|}{} &
  \multicolumn{1}{|c|}{8}     & \multicolumn{1}{|c|}{16}    &
  \multicolumn{1}{|c|}{32}    & \multicolumn{1}{|c|}{64}    &
  \multicolumn{1}{|c||}{128}   &
  \multicolumn{1}{|c|}{8}     & \multicolumn{1}{|c|}{16}    &
  \multicolumn{1}{|c|}{32}    & \multicolumn{1}{|c|}{64}    &
  \multicolumn{1}{|c|}{128}   \\ \hline
1M & 0.395 & 0.222  & 0.128  & 0.077  & 0.133  & 0.413 & 0.222 & 0.127 & 0.075  &0.135 \\ \hline
4M & 1.574 & 0.869  & 0.480  & 0.280  & 0.270  &1.807  & 0.985 & 0.514 & 0.294  &0.280\\ \hline
8M & 3.263 & 1.728  & 0.947  & 0.532  & 0.339  &3.967  &1.988  & 1.065 & 0.565	&0.352 \\ \hline
  \multicolumn{1}{|c|}{}      &
  \multicolumn{5}{|c||}{[RSR]} &
  \multicolumn{5}{|c|}{[RSQ]} \\ \cline{2-11}
  \multicolumn{1}{|c|}{} &
  \multicolumn{1}{|c|}{8}     & \multicolumn{1}{|c|}{16}    &
  \multicolumn{1}{|c|}{32}    & \multicolumn{1}{|c|}{64}    &
  \multicolumn{1}{|c||}{128}   &
  \multicolumn{1}{|c|}{8}     & \multicolumn{1}{|c|}{16}    &
  \multicolumn{1}{|c|}{32}    & \multicolumn{1}{|c|}{64}    &
  \multicolumn{1}{|c|}{128}   \\ \hline
1M &0.399  & 0.223  & 0.126  & 0.079  & 0.060  &0.416  &0.223  &0.124  & 0.076 &0.057 \\ \hline
4M & 1.593 &0.877   & 0.484  & 0.270  & 0.165  &1.826  &0.992  &0.516  & 0.284 & 0.164 \\ \hline
8M &3.161  & 1.746  & 0.957  &0.527   & 0.302  &3.915  & 2.00  &1.074  &0.558  &0.314 \\ \hline
\end{tabular}
\caption{{\small Scalability of {\sf [DSR], [RSR]} and {\sf [DSQ],[RSQ]} 
on same input [U].}}
\label{tbl11d}
}
\end{table}

Table \ref{tbl11d} depicts the scalability of the four sorting variants
introduced in this work
with increasing processor sizes for three moderate inputs sizes. As mentioned
earlier for $p=128$ and $n=8M$, the asymptotic claims for
{\sc Sort\_IRan\_BSP} and {\sc Sort\_Det\_BSP} hardly hold, and this
is also evidenced by the deteriorating performance of the
implementations for $n=4M$ and $n=1M$.
The slowdown for $n=1M$ is also attributable to the effect of the
extra code for duplicate handling;  disabling  of this piece of code
results in the elimination of the slowdown
observed for $n=1M$; the obtained speedup is however minimal.
We note that the experiments reported in Table \ref{tbl11d} 
constitute a separate set of experiments in addition to those
performed for the creation of Table \ref{tbl9}, Table \ref{tbl1},
and Table \ref{tbl2}; this explains slight differences in running
times which for almost all cases affect the third decimal digit of
various reported timing entries.

Another work that includes experimental results on deterministic
sorting appears in \cite{Hilla}. Table \ref{tbl12} below compares
{\sf [DSQ]} to a direct implementation of the regular
sampling algorithm of \cite{SH} reported in \cite{Hilla}.
The implementation  of \cite{Hilla} is similar to the deterministic
algorithm of \cite{jaja3} 
and the performance of the two algorithms is similar and significantly 
worse than the more refined algorithm {\sc Sort\_Det\_BSP} 
({\sf [DSQ]}) or the one in \cite{jaja1}. In addition, the 
algorithm in \cite{Hilla} as well as the algorithm in \cite{jaja3} 
can not handle duplicate keys.

\begin{table}[htb] \centering
{\small
\begin{tabular}{|l|r|r|r|r|r|r|r|} \hline 
Algorithm    & Input     & $n$           &{$p=8$}  &{$p=16$ }&{$p=32$ }&{$p=64$ }&{$p=128$} \\ \hline
{\sf [DSQ]}  &{\sf [U]}&$1024\times 1024$&0.413     &0.222    &0.127    &0.075    &0.135 \\ \hline
\cite{Hilla} &{\sf [U]}&$1000000$        & 0.462     &0.240    &0.137    &0.117    &-    \\ \hline
\end{tabular}
\caption{{\small Comparison of {\sf [DSQ]} on {\sc [U]} with 
$n=1024\times 1024$ keys to \cite{Hilla} on {\sc [U]} with $n=1,000,000$ keys.
}}
\label{tbl12}
%Table 9
}
\end{table}

%NOW
%The bitonic-sort algorithm, as demonstrated in the timing results of
%table~\ref{Tb:Cray3B}, performs relatively well for small processor
%numbers in accordance to the theoretical bound.
%In particular, the algorithm is competitive to {\em [DSD]} and {\em [RSD]}
%up to four processors and for small problem instances. For this reason,
%bitonic-sort has been employed in algorithms {\em [DSD]} and {\em [RSD]}
%as the underlying sample-sort method.
%As noted in \cite{jaja1}, despite the enormous theoretical interest in
%parallel sorting, the number of empirical studies is relatively small.
%All but three of these empirical studies used machines which are no longer
%available. The three exceptions are the recent work of \cite{jaja1,jaja2,SH96}.
%Figures~\ref{Fg:Jaja1}, \ref{Fg:Jaja2} and~\ref{Fg:SH96} compare 
%the performance of these studies with that of the deterministic and 
%randomized sample sort algorithms.  
%We note that in \cite{jaja1,jaja2} the code is written in {\sc Split-C} and
%makes use of high-performance MPI-like primitives.  On the other hand, 
%in \cite{SH96} the code is written in ANSI C and makes use of the Oxford
%BSP Toolset.

\section{Conclusion}
\label{s7}

In this work we have introduced deterministic and randomized algorithms
for internal memory sorting that were designed and their performance
analyzed in an architecture independent way under the BSP model.
The deterministic algorithm  \cite{GS99a} is an improvement of the 
regular sampling algorithm of \cite{SH} and uses the technique of 
deterministic oversampling (used for  randomized sorting in \cite{Reif87})
and in the general case, its theoretical performance is fully analyzed
in \cite{GS99a}. Some other of its advantages are its parallel sample-sort
that allows $p$ to be much close to $n$ than the algorithm in \cite{SH}.
From our experimental observations it seems
that parallel sample-sort is effective even for small problem sizes,
even if other
parallel sorting algorithms and implementations do not use it.
In addition, the introduced in this work transparent and efficient
duplicate handling method of Section \ref{duplicate} allows the algorithm to
handle duplicate keys without doubling the computation load or the
communication time that other approaches seem to 
require \cite{jaja1,jaja2,jaja3}.
The randomized algorithm although oversample-sort based
works differently from other sample-sort based approaches
and even improves upon the performance of the first
randomized BSP sorting algorithm of \cite{GV94}; it achieves this
by using ideas derived from  the deterministic algorithm. 
It also handles duplicate keys
with insignificant degradation in performance (computation or communication).
Given the insistence under the BSP model of one-optimality and the
accurate accounting of key imbalance during sorting, it is
easy to fine-tune the work-load balance and communication time by adjusting
the oversampling parameter and sample size of any of the two algorithms.
The experimental results presented in this work verify the theoretical claims
on the efficiency of the introduced algorithms. We have compared our
implementations to other parallel sorting implementations on the same
or similar platforms and realized that in terms of parallel
communication efficiency our algorithms seem to outperform other 
implementations and overall exhibit comparable if not better
performance (e.g. deterministic algorithm) even though the sequential
algorithms used in our implementations seem, in many
instances,  to be slower than those of other implementations. This suggests 
that one can probably improve overall performance of our implementations 
if sequential 
methods are optimized given the fact that in our communication efficient
designs sequential code takes 80-90\% of execution time for the problem 
instances examined.
As far as communication is concerned there is probably no room for
significant
improvement unless one tries to optimize the single
communication round by utilizing platform specific communication
primitives at the expense of portability.

The experimental portion of this work  was performed while both authors 
were with the Oxford University Computing Laboratory, The University 
of Oxford, Oxford, UK.  The work of the first author there was supported 
in part by EPSRC under grant GR/K16999 and that of the second author was 
supported in part by a Bodossaki Foundation Graduate Scholarship.
The support of the Edinburgh Parallel Computing
Centre in granting the first author access to a Cray T3D computer is 
gratefully acknowledged.
The work of the first author was subsequently supported in part
by NJIT SBR grant 421350 and NSF grant NSF-9977508.

\end{document}